\documentclass[a4paper,USenglish,cleveref,thm-restate]{lipics-v2021}

\title{An Automata-Based Approach to Games with \texorpdfstring{$\omega$}{ω}-Automatic Preferences}

\author{Véronique Bruyère}{Université de Mons (UMONS), Belgium \and \url{https://informatique-umons.be/bruyere-veronique/}}{veronique.bruyere@umons.ac.be}{https://orcid.org/0000-0002-9680-9140}{}

\author{Emmanuel Filiot}{Université libre de Bruxelles (ULB), Belgium \and \url{https://verif.ulb.ac.be/efiliot/}}{efiliot@ulb.be}{https://orcid.org/0000-0002-2520-5630}{is a Research Director of the Fonds de la Recherche
Scientifique – FNRS.}

\author{Christophe Grandmont}{Université de Mons (UMONS), Belgium \and Université libre de Bruxelles (ULB), Belgium \and \url{https://chrisgdt.github.io/}}{christophe.grandmont@umons.ac.be}{https://orcid.org/0009-0009-4573-0123}{}

\author{Jean-François Raskin}{Université libre de Bruxelles (ULB), Belgium \and \url{https://verif.ulb.ac.be/jfr/}}{jean-francois.raskin@ulb.be}{https://orcid.org/0000-0002-3673-1097}{Supported by Fondation ULB (\url{https://www.fondationulb.be/en/}) and the Thelam Foundation.
}

\authorrunning{V.\ Bruyère, E.\ Filiot, C.\ Grandmont, and J.-F.\ Raskin}


\newtheorem{problem}[theorem]{Problem}


\nolinenumbers
\pdfoutput=1
\hideLIPIcs

\funding{This work was supported by the Fonds de la Recherche Scientifique – FNRS under Grant Nos. T.0023.22 (PDR Rational) and T011724F (PDR Synthex).}

\acknowledgements{We thank the anonymous reviewers for their insightful comments.}

\Copyright{Véronique Bruyère, Emmanuel Filiot, Christophe Grandmont, and Jean-François Raskin}

\usepackage{tikz}
\usepackage{centernot}

\newtheorem*{problems*}{Problems}
\crefname{enumi}{Condition}{Conditions}

\mathchardef\mhyphen="2D

\newcommand{\<}{\langle}
\renewcommand{\>}{\rangle}
\renewcommand{\|}{\upharpoonright}

\newcommand{\ssetminus}{\! \setminus \!}
\newcommand{\val}{\mathsf{Val}}
\newcommand{\cval}{\mathsf{cVal}}

\newcommand{\aut}[1]{\ensuremath{\mathcal{#1}}}

\newcommand{\lang}{\mathcal{L}}

\newcommand{\NBW}{\text{NBW}}

\newcommand{\DBW}{\text{DBW}}

\newcommand{\ABW}{\text{ABW}}

\newcommand{\DPW}{\text{DPW}}

\newcommand{\DPWs}{\text{DPWs}}

\newcommand{\APW}{\text{APW}}
\newcommand{\APWs}{\text{APWs}}

\newcommand{\DFA}{\text{DFA}}
\newcommand{\DFAs}{\text{DFAs}}

\newcommand{\PO}{\mathsf{Pareto}}

\newcommand{\R}{\ensuremath{\mathrel{\ltimes}}}
\newcommand{\notR}{\ensuremath{\mathrel{\centernot\ltimes}}}

\newcommand{\arena}{A}

\newcommand{\game}{\mathcal{G}}

\newcommand{\players}{\mathcal{P}}

\newcommand{\plays}{\mathsf{Plays}}

\newcommand{\NEout}{\mathsf{NE}_{\text{out}}}

\newcommand{\bsigma}{\bar{\sigma}}

\newcommand{\outcome}[1]{\<#1 \>}
\newcommand{\outcomefrom}[2]{\outcome{#1}_{#2}}

\newcommand{\bigO}[1]{\mathcal{O}(#1)}

\newcommand{\nl}{$\mathsf{NL}$}
\newcommand{\nlHard}{\nl{}-hard}
\newcommand{\nlComplete}{\nl{}-complete}

\newcommand{\parity}{$\mathsf{Parity}$}
\newcommand{\parityHard}{$\mathsf{Parity}$-hard}
\newcommand{\parityComplete}{$\mathsf{Parity}$-complete}

\newcommand{\pspace}{$\mathsf{PSPACE}$}
\newcommand{\pspaceHard}{\pspace{}-hard}
\newcommand{\pspaceComplete}{\pspace{}-complete}

\newcommand{\twoExptime}{$\mathsf{2EXPTIME}$}

\usetikzlibrary{
  automata,
  arrows,
  positioning,
  shapes.geometric,
  calc,
  patterns
}
\tikzset{
edge with arrows/.style = {
    ->,
    >=stealth,
    shorten >=1pt,
},
directed/.style = {
    edge with arrows,
    node distance=2.3cm,
    on grid,
    semithick,
    double distance=1.5pt,
},
automaton/.style = {
    directed,
    auto,
    initial text={},
    pin distance = 1ex,
    every pin edge/.style = {
        draw=none
    },
    every state/.style={
      minimum size=1.5mm,
      inner sep=1pt,
    }
},
system/.style = {
    state,
    circle,
    minimum size=0mm,
    inner sep=5pt,
},
system2/.style={
    state,
    ellipse,
    minimum size=0mm,
    inner sep=2pt,
},
environment/.style = {
    state,
    rectangle,
    minimum size=0mm,
    inner sep=8pt,
},
environment2/.style = {
    state,
    diamond,
    minimum size=0mm,
    inner sep=4pt,
},
}

\begin{document}

\maketitle

\begin{abstract}
    This paper studies multiplayer turn-based games on graphs in which player preferences are modeled as $\omega$-automatic relations given by deterministic parity automata. This contrasts with most existing work, which focuses on specific reward functions. We conduct a computational analysis of these games, starting with the threshold problem in the antagonistic zero-sum case. As in classical games, we introduce the concept of value, defined here as the set of plays a player can guarantee to improve upon, relative to their preference relation. We show that this set is recognized by an alternating parity automaton (\APW{}) of polynomial size. We also establish the computational complexity of several problems related to the concepts of value and optimal strategy, taking advantage of the $\omega$-automatic characterization of value. Next, we shift to multiplayer games and Nash equilibria, and revisit the threshold problem in this context. Based on an \APW{} construction again, we close complexity gaps left open in the literature, and additionally show that cooperative rational synthesis is \pspaceComplete{}, while it becomes undecidable in the non-cooperative case.
\end{abstract}

\section{Introduction}
\label{section:intro}

Game theory on graphs has emerged as a foundational framework for the formal verification and synthesis of reactive systems~\cite{Games-on-Graphs,lncs2500,Pnueli-Rosner-1989}. Whereas the classical paradigm of \emph{zero-sum games} presupposes a strict antagonism between a system and an adversarial environment, contemporary applications---such as autonomous robotic fleets or distributed network protocols---demand models that capture interactions among agents with individual, not necessarily antagonistic, and frequently heterogeneous, objectives. This evolution motivates the analysis of \emph{non-zero-sum games}, in which agents (or players) are assumed to act rationally, a behavior classically formalized via solution concepts such as \emph{Nash equilibria} (NEs)~\cite{Nash50,OsborneRubinstein1994}.

Traditionally, the analysis of such games is divided into two distinct frameworks: the \emph{qualitative} setting, in which objectives are Boolean (typically $\omega$-regular conditions such as parity conditions or LTL specifications), and the \emph{quantitative} setting, in which agents seek to maximize numerical payoffs (for instance, mean-payoff or discounted-sum objectives)~\cite{ChatterjeeDHR10,Ehrenfeucht-Mycielski-1979}. Although the \emph{threshold problem}~\cite{vonNeumannMorgenstern1953}---that is, deciding whether a player can guarantee a given payoff---is by now well understood for these classes of objectives~\cite{Bruyere17}, the landscape becomes significantly fragmented when considering complex or hybrid winning conditions. Existing approaches are frequently tailored to specific numerical payoff functions or assume that the underlying preference relation forms a total preorder~\cite{BouyerBMU15,LeRoux-Pauly-Equilibria}, which limits their applicability to more general settings.

To provide a unifying viewpoint, recent work has advocated the use of \emph{automata-based preference relations}~\cite{BruyereGrandmontRaskin25-mfcs}. In this framework, player preferences are defined relationally rather than numerically: they are \emph{$\omega$-automatic relations}~\cite{lncs2500,BookSakarovitch}, i.e., relations specified by deterministic parity automata (\DPWs{}) that synchronously read pairs $(x,y)$ of infinite paths in the underlying game graph and accept them whenever $y$ is preferred to $x$, denoted by $x \R y$. As explained in \cite[end of Section 2]{BruyereGrandmontRaskin25-mfcs}, this framework captures many cases from classical games, as it encapsulates all $\omega$-regular objectives and several quantitative objectives (e.g., quantitative reachability, limsup, or discounted-sum, with an integer discount factor). The generality of this framework allows both (1) \emph{flexibility} and (2) \emph{modularity}:

(1) As \DPWs{} are closed under Boolean operations, our framework captures complex
combinations of classical qualitative criteria. Indeed, a player can combine several objectives possibly of different types (such as a Büchi objective with a quantitative reachability one). Moreover, they can combine them using different orders (such as a lexicographic order or a component-wise order). Crucially, this relational
approach does not impose the standard assumption that preferences must form a total preorder. This is already the case with the partial component-wise order, but also when preferences over three or more options are determined by multiple criteria. When aggregating distinct metrics such as cost, performance, and reliability, pairwise trade-offs can conflict across a set of alternatives to yield non-transitive preferences, which our relational framework naturally models without requiring structural constraints. Preferences may also depend on the history of actions played so far,
e.g., a player may prefer reaching target $B$ over target $C$ if and only if target $A$ has been visited in the past. Again, such a dynamic can be captured by the state space of a DPW representing the preference relation. Additional examples are given below in \Cref{ex:reachability-examples}.

(2) Similar to model checking, where the system and the specification are modeled separately (e.g., as a Kripke structure and a temporal logic formula, respectively), we adopt a modular approach in this paper in which preference relations are specified independently of the game arena. In contrast, in classical multiplayer games, preferences are typically homogeneous and hardcoded into the arena (e.g., via priorities or numerical values). Decoupling the dynamics of the game from the preference relations naturally allows for compositional representations and can yield exponentially more succinct descriptions of players' preferences. Indeed, each player's preference relation is specified independently by its own automaton without needing to account for the preferences of the other players, an issue that is inherent in classical approaches where all preferences are embedded in the arena.

This paper substantially extends the line of research on $\omega$-automatic preference relations, initiated in \cite{BruyereGrandmontRaskin25-mfcs}, where the following results were established: the NE existence problem is in \twoExptime{} and \pspaceHard{}, checking whether a play is an NE outcome is \parityComplete{}, and checking whether a given finite-memory strategy profile is an NE is \pspaceComplete{}. While NEs have been extensively studied for specific classes of objectives, especially the NE existence problem, a general and computationally efficient characterization of equilibrium outcomes for automata-based preference relations is lacking. To address this limitation, we revisit in this paper the analysis of zero-sum and multiplayer games with arbitrary $\omega$-automatic preference relations. Our approach relies on a generalized notion of \emph{value} that captures what a player can guarantee against an opponent, and serves as a foundation for studying NEs and related synthesis problems. This tightens the complexity bounds for NE existence (possibly with constraints) from \twoExptime{} to \pspaceComplete{}. Moreover, this notion of value allows us to derive important new results about cooperative and non-cooperative rational synthesis, as introduced in~\cite{FismanKL10,KupfermanPV16}.

\subparagraph*{Contributions}
We begin by revisiting the notion of \emph{value} for \emph{two-player zero-sum} games endowed with $\omega$-automatic preference relations $\R$: player~$1$ uses preference relation $\R$ while player~2 uses the reverse relation. In this enriched setting, the value that player~1 can enforce is no longer a single scalar quantity, but rather a set of infinite paths that they can guarantee to improve upon with respect to $\R$. We first prove that this value set is $\omega$-regular and can be recognized by an alternating parity automaton of polynomial size (\cref{thm:values-automatic}). We then extend the notion of \emph{optimal strategy} to this framework. These results allow us to characterize the computational complexity of several classical decision problems. This includes deciding the threshold problem for a given lasso path, deciding whether a threshold with a positive answer exists at all (\cref{thm:threshold-problem}), verifying the existence of an optimal strategy or finite-memory strategy optimality (\cref{thm:optimal-strategy}).

Building on these zero-sum results, we then analyze \emph{NEs} in \emph{multiplayer} games. We show that the set of NE outcomes can be characterized as the intersection of the \emph{consistent values} (a variant of the zero-sum value) of the coalitions opposed to each player, and thus recognized by an alternating parity automaton of polynomial size (\cref{thm:apw-for-ne-outcome}). This characterization leads to an automata-based procedure for checking NE existence and allows us to establish that the NE existence problem, with or without constraints (given as lasso paths), is \pspaceComplete{} (\cref{thm:NEthreshold}). This improves the known \twoExptime{} upper bound of~\cite{BruyereGrandmontRaskin25-mfcs} and closes this significant complexity gap. We also get \pspaceComplete{}ness for constraints specified by $\omega$-regular or LTL conditions~(\cref{thm:LTLconstraint}).

Finally, we investigate the problem of \emph{rational synthesis}, originally introduced in the classical (qualitative) setting in~\cite{FismanKL10,KupfermanPV16}. This problem asks whether one can synthesize a strategy for a system in a way that satisfies a specification in an environment composed of rational agents whose joint behavior forms an NE in response to the system's strategy. A distinction is made between \emph{cooperative} rational synthesis, where the synthesized strategy is required to realize the specification under the assumption that the rational environment collaboratively selects a favorable equilibrium, and \emph{non-cooperative} rational synthesis, where the specification must be satisfied for \emph{all} induced equilibria.
In our general setting of games with $\omega$-automatic preference relations, we prove that cooperative rational synthesis is \pspaceComplete{}, while non-cooperative rational synthesis becomes undecidable, when the specification is a lasso path constraint (\cref{thm:crs-pspace-ncrs-undec}). This undecidability result is particularly striking and unexpected; it exhibits a sharp complexity-theoretic separation from the classical qualitative case~\cite{ConduracheFGR16}. We further consider the associated \emph{verification problems}: given a strategy for the system defined as a Mealy machine, we verify whether the specification is guaranteed to be satisfied against a rational environment. Specifically, cooperative (resp.\ non-cooperative) rational verification asks whether the specification holds for at least one (resp.\ all) NEs formed by the environment's agents in response to this fixed strategy. We show that both verification problems are \pspaceComplete{} (\cref{thm:verification-rational-synthesis}).

\subparagraph*{Related work}
The formal study of synchronized relations on words is grounded in the theory of rational relations~\cite{BookSakarovitch} and $\omega$-automatic structures, which have been extensively developed for applications in logic (e.g., first-order logic)~\cite{BlumensathGradel2000,Gradel2020,Hodgson1983,KhoussainovNerode1995}. The structural and algorithmic properties of these relations have been the subject of dedicated research~\cite{Bergstrasser-Ganardi-2023,BruyereGrandmontRaskin25-mfcs,KaiserRubinBarany-2008}. In the latter context, automata-based preference relations for games were recently explored in~\cite{BruyereGrandmontRaskin25-mfcs}, while related approaches link $\omega$-automatic aggregate functions to quantitative objectives~\cite{BansalChaudhuriVardi2022,Rajasekaran-Bansal-Vardi-2023}. Furthermore, automatic relations have proven effective for modeling imperfect information in games~\cite{Berwanger-Doyen-2023,Bozzelli-Maubert-Pinchinat-2015}. The algorithmic treatment of such expressive specifications relies here on the theory of alternating automata, whose origins~\cite{ChandraKozenStockmeyer1981,MiyanoHayashi1984} and applications (e.g., to linear temporal logics~\cite{KupfermanS06,Vardi95}) provide essential tools for algorithmic studies.

In the domain of zero-sum games, the analysis of quantitative objectives is well-established \cite{Bruyere17,ChatterjeeDHR10,Ehrenfeucht-Mycielski-1979}. The threshold problem, in particular, has been widely studied for both Boolean and classical quantitative objectives~\cite{BouyerBMU15,BrihayeBGT21,BrihayePS13,Bruyere17,Jurdzinski98}. Research has extended to preference relations with an ordered structure~\cite{BouyerBMU15,BruyereGrandmontRaskin25-mfcs,LeRoux-Pauly-Equilibria}. Moving to the non-zero-sum setting, rational behavior is captured by NEs~\cite{Nash50}. The NE threshold problem is studied for specific reward functions~\cite{Bouyer-Brenguier-Markey-2010,Bruyere17,ChatterjeeHJ06,Gutierrez-MPRSW21,Ummels08,UmmelsWojtczak2011-meanpayoff} but a general theory is still missing. Recent work has examined reward functions associating integers with subsets of satisfied objectives of the same type~\cite{Feinstein-Kupferman-Shenwald2025,KupfermanPV16}. Automata-theoretic approaches have also been employed to characterize NE outcomes in quantitative settings~\cite{KlimosLST12,Rajasekaran-Bansal-Vardi-2023}, and the set of outcomes compatible with admissible strategies~\cite{BrenguierRS14}. In the context of $\omega$-recognizable relations satisfying specific hypotheses (total preorder, prefix-independence), NE outcomes were characterized in~\cite{BruyereGrandmontRaskin25-mfcs} via a restricted notion of value. Finally, the problem of rational synthesis, originally developed in~\cite{FismanKL10}, has attracted significant attention~\cite{AlmagorKupfermanPerelli2018,Bruyere21,GutierrezNPW23,Kupferman-Shenwald-2022} and has been extended to quantitative objectives~\cite{BriceRB23-rational-verif,BruyereGrandmontRaskin24-concur}.

\section{Preliminaries}
\label{section:preliminaries}

In this section, we recall the concept of multiplayer games with $\omega$-automatic preference relations as introduced in~\cite{BruyereGrandmontRaskin25-mfcs}, as well as several useful definitions and an illustrative example.

We assume that the reader is familiar with the classical notion of parity and B\"uchi automata over infinite words. We use the classical terminology ``XYW'' where $X \in \{D,N,A\}$ is either deterministic (D), nondeterministic (N), or alternating (A), and $Y \in \{B,P\}$ is either B\"uchi (B) or parity (P).\footnote{$W$ stands for infinite Words.} For example, a deterministic parity automaton is denoted \DPW{}. For a parity automaton, we denote by $\alpha$ its priority function that assigns a priority in $\{0,1,\ldots,k\}$ to each state, and we denote its index by $|\alpha|$, equal to the maximum priority~$k$. A run is accepting if the maximum priority seen infinitely often is even. For a B\"uchi automaton, we denote by $F$ the set of accepting states that a run needs to visit infinitely often to be accepted. See, e.g.,~\cite{lncs2500,handbook-automata-kupferman18,handbook-Automata-WilkeSchewe21}, for more details. Given an automaton $\aut{A}$, we denote by $|\aut{A}|$ its size, the number of its states.

We consider binary relations $\mathord{\R} \subseteq \Sigma^\omega \times \Sigma^\omega$ on infinite words over a finite alphabet $\Sigma$. A relation $\R$ is \emph{$\omega$-automatic} if the language $L_{\R} = \{ (a_1,b_1)(a_2,b_2)\dots\in(\Sigma\times\Sigma)^\omega\mid (a_1a_2\dots,b_1b_2\dots)\in \mathop{\R}\}$ is accepted by a \DPW{} over the alphabet $\Sigma \times \Sigma$, that is, $L_{\R}$ is an $\omega$-regular language~\cite{BookSakarovitch}. The automaton reads pairs of letters by advancing synchronously on the two words. In the sequel, we write $x \R y$ (resp.\ $x \notR y$) instead of $(x,y) \in \mathop{\R}$ (resp.\ $(x,y) \not\in \mathop{\R}$).

An \emph{arena} is a tuple $\arena = (V,E,\players,(V_i)_{i\in\players})$ where $V$ is a finite set of vertices, $E \subseteq V \times V$ is a set of edges, $\players$ is a finite set of players, and $(V_i)_{i\in\players}$ is a partition of $V$, with $V_i$ the set of vertices owned by player~$i$. We assume, w.l.o.g., that each $v \in V$ has at least one successor, i.e., there exists $v' \in V$ such that $(v,v') \in E$. A \emph{play} $\pi\in V^\omega$ (resp.\ a \emph{history} $h\in V^*$) is an infinite (resp.\ finite) sequence of vertices $\pi_0\pi_1\dots$ such that $(\pi_k,\pi_{k+1})\in E$ for all $k$. The set of all plays is denoted $\plays$, and we write $\plays(v)$ for the set of plays starting with vertex $v$. A play $\pi= \mu\nu^\omega$, where $\mu\nu$ is a history and $\nu$ an non-empty cycle, is a \emph{lasso}. Its \emph{length} $|\pi|$ is the number of vertices of~$\mu\nu$.\footnote{To have a well-defined length $|\pi|$, we assume that $\pi = \mu\nu^\omega$ with $\mu\nu$ as short as possible.}

A \emph{game} $\game = (\arena, (\R_i)_{i\in\players})$ is an arena equipped with $\omega$-automatic relations $\R_i$ over the alphabet $V$ (each accepted by a \DPW{}), one for each player~$i$, called their \emph{preference relation}. For any two plays $\pi, \pi'$, player~$i$ \emph{prefers} $\pi'$ to $\pi$ if $\pi \R_i \pi'$.
In this paper, each $\R_i$ is just an $\omega$-automatic relation without any additional hypotheses. We denote by $|V|$ the size of a game $\game$, with $V$ the set of its vertices.

Let $\arena$ be an arena. A \emph{strategy} $\sigma_i:V^*V_i\rightarrow V$ for player~$i$ maps any history $hv \in V^*V_i$ to a successor of vertex $v$. The set of all strategies of player~$i$ is denoted~$\Gamma_i$. A strategy $\sigma_i$ is \emph{finite-memory} if it can be encoded by a Mealy machine $\aut{M}$~\cite{lncs2500}. The memory size of $\sigma_i$ is the size $|\aut{M}|$ of $\aut{M}$, equal to the number of its states. A strategy $\sigma_i$ is \emph{memoryless} if $\sigma_i(hv) = \sigma_i(v)$ for all histories $h$; it is a particular case of finite-memory strategy encoded by a Mealy machine $\mathcal{M}$ with size $|\aut{M}| = 1$. Given a strategy $\sigma_i \in \Gamma_i$ and a history $h$, we define the strategy $\sigma_{i\|h}$ from $\sigma_i$ by $\sigma_{i\|h}(h') = \sigma_i(hh')$ for all histories $hh' \in V^*V_i$.

A play $\pi = \pi_0\pi_1 \dots$ is \emph{consistent} with a strategy $\sigma_i$ if $\pi_{k+1} = \sigma_i(\pi_0 \dots \pi_k)$ for all $k$ such that $\pi_k \in V_i$.
Consistency is naturally extended to histories. Given a strategy $\sigma_i \in \Gamma_i$ and a vertex $v \in V$, we denote by $\plays_{\sigma_i}(v)$ the set of all plays starting at $v$ that are consistent with $\sigma_i$. A tuple of strategies $\bsigma = (\sigma_i)_{i\in\players}$ with $\sigma_i \in \Gamma_i$ for each player~$i$, is called a \emph{strategy profile}. The play $\pi$ starting from an initial vertex $v$ and consistent with each $\sigma_i$ is denoted by $\outcomefrom{\bsigma}{v}$ and called its \emph{outcome}.

We suppose that the reader is familiar with the classical concepts of two-player \emph{zero-sum} games with $\omega$-regular objectives and of winning strategy~\cite{Games-on-Graphs,lncs2500}.

Several examples are given in the introduction and in \cite[end of Section 2]{BruyereGrandmontRaskin25-mfcs} to demonstrate the flexibility and generality of $\omega$-automatic preference relations. In the next example, we further illustrate this framework by examining multiple variations of reachability objectives.

\begin{example}\label{ex:reachability-examples}
    Given an arena $\arena$, we first consider \emph{min-cost reachability}, where the objective is to reach a target set $T \subseteq V$ while minimizing the number of vertices. Given $x \in V^\omega$, let $c(x)$ denote the minimal number of vertices to reach the target $T$ along $x$ (with $c(x) = +\infty$ if the target is never reached). We define the relation $\R$ such that for all $x, y \in V^\omega$, $x \R y$ if $c(x) \geq c(y)$. A \DPW{} accepting $\R$ is depicted in \cref{fig:min-cost-reachability}, such that the label, e.g., ``$(T,\neg T)$'' represents the set of labels $(v,u)$ with $v \in T$ and $u \not\in T$, and the numbers inside the states represent their priority.
    The relation $\R$ induces a total preorder on $V^\omega$. A strict order can be obtained by using strict inequality $c(x) > c(y)$.

    Second, we consider \emph{max-reward reachability}, where the objective is to maximize the number of visited vertices to reach the target $T$. The relation $\R$ is then defined as $x \R y$ if $(c(x) \leq c(y) < +\infty) \lor c(x) = +\infty$, and a \DPW{} accepting $\R$ is depicted in \cref{fig:max-reward-reachability}. This relation induces a total preorder that is the inverse of the preorder for min-cost reachability (except regarding $+\infty$, where failing to reach the target is the worst-case scenario). Missing transitions are assumed to lead to a rejecting sink state.

    Finally, we shift to partial preorders, with a combination of \emph{two Boolean reachability} objectives. Given two targets $T_1, T_2\subseteq V$, we define the relation $\R$ as $x \R y$ if for all $i \in \{1,2\}$, if $x$ visits $T_i$ then $y$ visits $T_i$. Since $x$ visiting only $T_1$ and $y$ visiting only $T_2$ are incomparable with respect to $\R$, we get a partial preorder. The state space of a \DPW{} $\aut{A}$ accepting $\R$ is represented in \cref{fig:lattice_automaton}, as a grid where the vertical (resp.\ horizontal) axis corresponds to $x$ (resp.\ $y$) and indicates the visited targets among $T_1$ and $T_2$. States marked with ``$\checkmark$ Accept'' (in green) have priority $0$ and states marked with ``$\times$ Reject'' (in red) have priority $1$. Thus, $x \R y$ holds exactly when the automaton ultimately stays in an accepting state while reading $(x,y)$.
    \lipicsEnd
\end{example}

\begin{figure}[t]
    \centering
    \begin{minipage}[c]{0.56\textwidth}
        \centering
        \begin{minipage}[t]{0.48\linewidth}
            \centering
            \begin{tikzpicture}[automaton,scale=.7,every node/.style={scale=.7},node distance=1.5]
                \node[initial below,system] (q0) {$0$};
                \node[system] (q1) [right=of q0] {$0$};
                \node[system] (q2) [left=of q0] {$1$};
                
                \path (q0) edge            node {$(*, T)$} (q1)
                           edge            node[above] {$(T, \neg T)$} (q2)
                           edge[loop above]node {$(\neg T, \neg T)$} (q0)
                      (q1) edge[loop above]node {$(*, *)$} (q1)
                      (q2) edge[loop above]node {$(*, *)$} (q2);
            \end{tikzpicture}
            \caption{A \DPW{} accepting the relation $\R$ for min-cost-reachability with  target $T$.}
            \label{fig:min-cost-reachability}
        \end{minipage}
        \hfill
        \begin{minipage}[t]{0.48\linewidth}
            \centering
            \begin{tikzpicture}[automaton,scale=.65,every node/.style={scale=.65},node distance=1.55]
              \node[initial,system] (q0) {$0$};
              \node[system] (q1) [below=of q0] {$1$};
              \node[system] (q2) [right=of q1] {$0$};
              \node[system] (q3) [right=of q0] {$0$};
              
              \path (q0) edge            node[left] {$(T, \neg T)$} (q1)
                         edge            node {$(\neg T, T)$} (q3)
                         edge[loop above]node {$(\neg T, \neg T)$} (q0)
                         edge            node {$(T, T)$} (q2)
                    (q1) edge[loop left] node {$(*, \neg T)$} (q1)
                         edge            node {$(*, T)$} (q2)
                    (q2) edge[loop above]node {$(*, *)$} (q2)
                    (q3) edge[loop above]node {$(\neg T, *)$} (q3);
            \end{tikzpicture}
            \caption{A \DPW{} accepting the relation $\R$ for max-reward reachability.}
            \label{fig:max-reward-reachability}
        \end{minipage}

        \vspace{1em}

        \begin{minipage}[t]{1.0\linewidth}
            \centering
            \begin{tikzpicture}[automaton,scale=.65,every node/.style={scale=.65}]
                \node[system] (v0) {$v_0$};
                \node[environment] (v1) [above right=.5cm and 1cm of v0] {$v_1$};
                \node[environment] (v2) [below right=.5cm and 1cm of v0] {$v_2$};

                \path (v0) edge[loop above] (v0)
                           edge (v1)
                           edge (v2)
                      (v1) edge[loop right] (v1)
                           edge[bend right=20] (v2)
                      (v2) edge[loop right] (v2)
                           edge[bend right=20] (v1);
            \end{tikzpicture}
            \addtocounter{figure}{1}
            \caption{An arena with round (resp.\ square) vertices owned by player~$1$ (resp.\ player~$2$).}
            \label{fig:example-optimal-strat}
        \end{minipage}
    \end{minipage}
    \hfill
    \begin{minipage}[c]{0.4\textwidth}
        \centering
        \begin{tikzpicture}[>=stealth,scale=.95,every node/.style={scale=.8}]
            \draw[thick] (0,0) -- (4,0) node[right] {$y$};
            \draw[thick] (0,0) -- (0,4) node[above] {$x$};

            \foreach \i/\lbl in {0/$\varnothing$, 1/$\{T_1\}$, 2/$\{T_2\}$, 3/{$\{T_1, T_2\}$}} {
                \node at (\i+0.5, -0.3) {\small \lbl};
                \node at (-0.5, \i+0.5) {\small \lbl};
            }

            \foreach \x in {0,1,2,3} {
                \foreach \y in {0,1,2,3} {
                     \pgfmathsetmacro{\isSafe}{
                        (\y==0) ? 1 : (
                            (\y==1 && (\x==1 || \x==3)) ? 1 : (
                                (\y==2 && (\x==2 || \x==3)) ? 1 : (
                                    (\y==3 && \x==3) ? 1 : 0
                                )
                            )
                        )
                     }

                    \ifnum\isSafe=1
                        \fill[green!8!white] (\x,\y) rectangle ++(1,1);
                        \node[green!40!black, font=\tiny, align=center] at (\x+0.5,\y+0.5) {$\checkmark$\\Accept};
                    \else
                        \fill[red!18!white] (\x,\y) rectangle ++(1,1);
                        \fill[pattern=north east lines, pattern color=red!30!white] (\x,\y) rectangle ++(1,1);
                        \node[red!40!black, font=\tiny, align=center, fill=red!18!white, inner sep=1pt] at (\x+0.5,\y+0.5) {$\times$\\Reject};
                    \fi
                }
            }
            \draw[step=1cm, gray!30, thin] (0,0) grid (4,4);
        \end{tikzpicture}
        \addtocounter{figure}{-2}
        \caption{A high-level representation of the states of a \DPW{} accepting the relation $\R$ for the combination of two Boolean reachability objectives.
        }
        \label{fig:lattice_automaton}
        \addtocounter{figure}{1}
    \end{minipage}
\end{figure}

While the previous example illustrates preference relations satisfying specific properties such as being a preorder, it is important to note that our framework \emph{does not impose} such structural assumptions. Unlike classical game-theoretic settings where preferences are typically modeled as strict orders or preorders, we define $\R$ as an arbitrary $\omega$-automatic relation. Therefore, this \emph{general framework} encompasses standard ordered objectives as well as their combinations (several examples are given in~\cite{BruyereGrandmontRaskin25-mfcs}), but also captures non-standard relations. Consequently, all subsequent concepts studied in this paper, such as the values and Nash Equilibria, remain operational without any additional hypotheses. The intrinsic meaning of ``preference'' relies solely on the definition of $\R$: it can represent a standard strict partial order or preorder, but also cyclic relations~\cite{BrihayeGHR20} or more exotic ones. If specific properties such as (ir)reflexivity or transitivity are required for a particular application, checking their satisfaction is known to be \nlComplete{}~\cite{BruyereGrandmontRaskin25-mfcs}. Note that for all three objectives of~\cref{ex:reachability-examples}, it is easy to slightly modify the \DPW{} accepting $\R$ such that it accepts the corresponding strict order defined as: $x \prec y$ if and only if $x \R y \land y \notR x$.

\section{Two-Player Zero-Sum Games}
\label{section:zero-sum}

In this section, we study the purely antagonistic setting of two-player zero-sum games. We have one preference relation $\R_1$ for player~$1$ and the reverse relation $\R_2$ for player~$2$: $x \R_1 y$ if and only if $y \R_2 x$~\cite{BouyerBMU15}. Intuitively, whenever one player perceives a play as an improvement over another, the opponent perceives it as a degradation. In this section, we denote by $\R$ the relation $\R_1$. We thus fix a two-player zero-sum game $\game = (\arena,\R)$ on the arena $\arena = (V,E,\{1,2\},(V_1,V_2))$, with $\R$ an $\omega$-automatic preference relation accepted by a \DPW{} $\aut{A} = (Q,q_0,\delta,\alpha)$. Our main goal is to study the following \emph{threshold problem}.

\begin{problem}[Threshold problem]\label{problem:threshold}
    Given a game $\game = (\arena, \R)$, a vertex $v\in V$, and a lasso $\pi \in \plays(v)$ (called threshold), the \emph{threshold problem} is to decide whether there exists a strategy $\sigma_1$ for player~$1$ such that for all plays $\rho \in \plays(v)$ consistent with $\sigma_1$, we have $\pi \R \rho$.
\end{problem}

A standard approach to solving the threshold problem is to introduce the notion of value. In game theory, it is defined as the optimal payoff that player~1 can guarantee against player~2. This concept was introduced as early as in the seminal work of von Neumann and Morgenstern~\cite{vonNeumannMorgenstern1953}. In classical settings, the value is typically a real scalar $c$ and player~1 can guarantee a payoff arbitrarily close to $c$, and even reach it in case of existence of an optimal strategy. Once the value $c$ is computed, it immediately yields a solution to the threshold problem: one simply verifies whether the specified threshold is less than $c$. This classical notion of value can equivalently be viewed not only as a real scalar $c$, but also as the set of all payoffs less than $c$, that is, the set of outcomes that player~1 can guarantee to improve upon. Hence, in our setting where no assumptions are imposed on $\R$ (except being $\omega$-automatic), we define the value in terms of the set of plays $\pi$ for which player~$1$ can guarantee an improvement. In this way, we eliminate the scalar $c$, and replace it by the set $\val_1(v)$ as presented in the next definition. Building on this alternative definition, an optimal strategy is one that guarantees an outcome preferable to any play in $\val_1(v)$. This is illustrated in the first case of \cref{ex:reachability-examples-sequel} below.

\begin{definition}\label{def:values}
Let $\game = (\arena,\R)$ be a two-player zero-sum game and $v \in V$ be a vertex.
\begin{itemize}
    \item The \emph{value of $v$} is the set $\val_1(v) = \{\pi \in \plays(v) \mid \exists \sigma_1 \in \Gamma_1, \forall \rho \in \plays_{\sigma_1}(v),\, \pi \R \rho\}$.
    \item An \emph{optimal strategy from $v$} for player~$1$ is a strategy $\sigma_1^* \in \Gamma_1$ such that for all $\pi \in \val_1(v)$ and $\rho \in \plays_{\sigma_1^*}(v)$, we have $\pi \R \rho$.
\end{itemize}
\end{definition}

With this definition, the threshold problem reduces to checking whether the given lasso $\pi$ belongs to $\val_1(v)$. While a non-optimal strategy may suffice as a solution to the threshold problem for a specific threshold, an optimal strategy $\sigma_1^*$, if it exists, acts as a universal witness: it provides a solution to the threshold problem for every lasso within the value set. Let us illustrate \cref{def:values} using the $\omega$-automatic relations $\R$ from \cref{ex:reachability-examples}.

\begin{example}\label{ex:reachability-examples-sequel}
    Let us first consider the relation $\R$ for min-cost reachability, accepted by the \DPW{} in \cref{fig:min-cost-reachability}. In this context, the threshold problem consists of deciding, given a lasso $\pi \in \plays(v)$, whether there exists a strategy $\sigma_1$ such that $\pi \R \rho$, for all $\rho$ starting from $v$ and consistent with $\sigma_1$.
    That is, a strategy that guarantees $c(\rho) \leq c(\pi)$. Let $c^*$ be the best threshold\footnote{That is, the lowest cost $c$.} $c$ that player~1 can guarantee; then the set $\val_1(v)$ is equal to $\{\pi \in \plays(v) \mid c^* \leq c(\pi)\}$. This problem is well-studied, $c^*$ is the classical notion of value, and an optimal strategy always exists~\cite{BrihayeGHM17}.

    Let us now consider the relation $\R$ for max-reward reachability. Unlike the previous case, an optimal strategy for player~$1$ does not always exist: see the arena $A$ depicted in \cref{fig:example-optimal-strat}, where player~$1$ owns $v_0$ and $T = \{v_1,v_2\}$. One can verify that $\val_1(v_0) = \plays(v_0)$. Any strategy $\sigma_1$ for player~$1$ either loops $k$ times on $v_0$ and then leaves $v_0$, or loops on $v_0$ indefinitely. In the first case, for any play $\rho$ consistent with $\sigma_1$, we have $\pi \not\R \rho$ with $\pi = (v_0)^{k+2} v_1^\omega$, showing that $\sigma_1$ is not optimal. In the second case, the unique play $\rho$ consistent with $\sigma_1$ is $v_0^\omega$, and again $\pi \not\R \rho$ with $\pi = v_0v_1^\omega$. Nevertheless, it should be noted that, even if there is no optimal strategy, the threshold problem still has a solution for every lasso $\pi \in \plays(v_0)$. Indeed, as $\val_1(v_0) = \plays(v_0)$, player~$1$ can always pick a strategy looping on $v_0$ sufficiently many times.

    Finally, let us consider the relation $\R$ for the combination of two Boolean reachability objectives. 
    We illustrate the absence of an optimal strategy for player~$1$ with the same arena of \cref{fig:example-optimal-strat}, and $T_i = \{v_i\}$, $i = 1,2$. One can verify that the value $\val_1(v_0)$ contains all plays that do not visit both vertices $v_1$ and $v_2$. A candidate optimal strategy $\sigma_1^*$ would need to guarantee that any consistent play $\rho$ satisfies $\pi_1 \R \rho$ and $\pi_2 \R \rho$, with $\pi_1 = v_0v_1^\omega$ and $\pi_2 = v_0v_2^\omega$, i.e., $\rho$ would have to visit both $v_1$ and $v_2$. However, player~2 can always prevent such a visit, so there exists no optimal strategy. Similar to the previous case, the threshold problem remains solvable for any specific lasso $\pi \in \val_1(v_0)$. Since any such threshold $\pi$ visits at most one of $T_1$ or $T_2$, player~$1$ can always visit the same target as $\pi$.
    \lipicsEnd
\end{example}

\begin{remark}
    In \cref{def:values}, one might be tempted to replace $\val_1(v)$ by its maximal plays. This is not correct in our general setting of preference relations where no conditions are imposed. Therefore, maximal elements may not exist, as shown in the previous example with max-reward reachability. Our set-based definition avoids this issue entirely, yielding a robust and well-defined notion of value for arbitrary preference relations.
    \lipicsEnd
\end{remark}

For the rest of the section, $\game = (\arena,\R)$ is a two-player zero-sum game, with arena $\arena = (V,E,\{1,2\},(V_1,V_2))$, such that the preference relation $\R$ is accepted by a \DPW{} $\aut{A}= (Q,q_0,\delta,\alpha)$ with a priority function $\alpha$. An instance of the threshold problem consists of $\game$, an initial vertex $v \in V$, and a lasso $\pi \in \plays(v)$.
A \emph{key result} is the following lemma stating that the value $\val_1(v)$ is an $\omega$-regular set.

\begin{lemma}\label{thm:values-automatic}
    Let $\game$ be a game and $v_0 \in V$ be a vertex. Then, the value $\val_1(v_0)$ is accepted by an \APW{} $\aut{B}$ with $1 + |V|^2 \cdot |\aut{A}|$ states and a priority function of index $|\alpha|$.
\end{lemma}

\begin{proof}
     We define an \APW{} $\aut{B}$ over the alphabet $V$ that accepts $\val_1(v_0)$ as follows. While reading a word $\pi\in V^\omega$, it uses existential transitions to guess a strategy for player~$1$ and universal transitions to check that whatever player~2's choices, the resulting play $\rho$ is preferred to $\pi$. To do so, it simulates the \DPW{} $\mathcal{A}$ over the synchronized product of $\pi$ and $\rho$. This is correct as $\mathcal{A}$ is deterministic.
     Formally, $\aut{B} = (S,s_{\text{init}},\gamma,\beta)$ is constructed with the set of states $S = (V \times V \times Q) \cup \{s_{\text{init}}\}$ where $s_{\text{init}}$ is the initial state, and the priority function is defined by $\beta(v,u,q) = \alpha(q)$ and $\beta(s_{\text{init}}) = 0$. From $s_{\text{init}}$, there is only one transition reading $v_0$ such that $\gamma(s_{\text{init}},v_0) = (v_0,v_0,q_0)$. From any state $(v,u,q)$ and letter $v' \in V$ such that $(v,v') \in E$, $\gamma((v,u,q),v')$ equals $\bigvee_{(u,u') \in E} (v',u',q')$ if $u \in V_1$, and $\bigwedge_{(u,u') \in E} (v',u',q')$ if $u \in V_2$, with $q' = \delta(q,(v,u))$.
     Hence, $\aut{B}$ has $1 + |V|^2 \cdot |\aut{A}|$ states and a priority function $\beta$ of index $|\alpha|$.
     
     Consider a run of $\aut{B}$ on a word $\pi \in V^\omega$. After verifying that $\pi$ starts with $v_0$, it forms a tree in which vertices belonging to $V_1$ (resp.\ $V_2$) have only one child (resp.\ have all their children), by definition of $\gamma$. Such a run can be seen as a strategy $\sigma_1$ for player~1. It is accepting if and only if all its branches satisfy the parity condition given by $\beta$. Consider a branch: on its first component, it is checked whether $\pi$ is a play starting at $v_0$; on its second component, another play $\rho \in \plays(v_0)$ consistent with $\sigma_1$ is explored; finally, on the third component, it is checked whether $\pi \R \rho$ by definition of $\beta$. In other words, the considered run is accepted by $\aut{B}$ if and only if all plays $\rho$ consistent with $\sigma_1$ satisfy $\pi \R \rho$, i.e., $\pi \in \val_1(v_0)$.
\end{proof}

From~\cref{thm:values-automatic}, we can derive the complexity classes of the threshold problem and the problem of deciding the existence of a threshold~$\pi$ for which a solution to the threshold problem exists. We also solve the related problem of verifying whether a finite-memory strategy is a solution to the threshold problem for a given threshold $\pi$. A problem being \parityComplete{} refers to being equivalent, i.e., polynomially inter-reducible, to solving a zero-sum parity game.

\begin{theorem}[restate=thmthresholdproblem,name=]
\label{thm:threshold-problem}
    Given a game $\game = (\arena, \R)$ and an initial vertex $v \in V$:
    \begin{itemize}
        \item Given a lasso $\pi \in \plays(v)$, the threshold problem is \parityComplete{}, and admits finite-memory strategies as solutions, with memory size $|\aut{A}| \cdot |\pi|$.
        \item Deciding whether there exists a lasso $\pi \in \plays(v)$ for which the threshold problem has a solution is \pspaceComplete{}.
        \item Given a lasso $\pi \in \plays(v)$ and a strategy $\sigma_1\in \Gamma_1$ defined by a Mealy machine, verifying whether $\sigma_1$ is a solution to the threshold problem for $\pi$ is \nlComplete{}.
    \end{itemize}
\end{theorem}

\begin{proof}[Proof (Sketch)]
    We only give sketches of the proofs for the \parity{} and \pspace{} memberships. The full proof is given in~\cref{app:thresholdproblems}.
    The threshold problem amounts to checking whether a given lasso $\pi$ belongs to $\val_1(v)$, i.e., to solving a two-player zero-sum parity game~\cite{handbook-Automata-WilkeSchewe21} obtained as the product of the lasso $\pi$ and the \APW{} \aut{B} of \cref{thm:values-automatic}.
    The second problem amounts to deciding whether $\val_1(v) \neq \varnothing$. The nonemptiness problem for an \APW{} is in \pspace{}~\cite{DaxKlaedtke08,Leucker-Sanchez-2010}: construct on the fly an \NBW{} of exponential size equivalent to $\aut{B}$, while verifying its non-emptiness with an \nl{} algorithm.
\end{proof}

\Cref{ex:reachability-examples-sequel} showed the existence of a game with no optimal strategy. Therefore, we study two related problems: verifying whether a finite-memory strategy is optimal, and deciding whether there exists an optimal strategy. Both problems can be solved, as consequences of \cref{thm:values-automatic}. Their proofs are given in~\cref{app:optimalstrategies}.

\begin{theorem}[restate=thmoptimalstrategy,name=]
\label{thm:optimal-strategy} 
    Given a game $\game = (\arena, \R)$ and an initial vertex $v \in V$:
    \begin{itemize}
    \item Given a strategy $\sigma_1 \in \Gamma_1$ defined by a Mealy machine, verifying whether $\sigma_1$ is an optimal strategy from $v$ is \pspaceComplete{}.
    \item The problem of deciding whether there exists an optimal strategy $\sigma_1^*$ from $v$ is in \twoExptime{} and \pspaceHard{}. The lower bound holds even for one-player games. In case of existence, there exists a finite-memory solution $\sigma_1^*$, with memory size doubly exponential in $|V|$, $|\aut{A}|$, and $|\alpha|$.
    \end{itemize}
\end{theorem}

\begin{remark}\label{rem:Val2}
In this section, we have presented the concept of value and optimal strategy for player~$1$. Those concepts can also be defined for player~$2$ (recall that the preference relation of player~$2$ is the reverse of $\R$): the value $\val_2(v)$ is the set $\{\pi \in \plays(v) \mid \exists \sigma_2 \in \Gamma_2, \forall \rho \in \plays_{\sigma_2}(v),\, \rho \R \pi\}$, and a strategy $\sigma_2^* \in \Gamma_2$ is optimal if for all $\pi \in \val_2(v)$ and $\rho \in \plays_{\sigma_2^*}(v)$, we have $\rho \R \pi$.
Notice that the reverse of $\R$ used in these definitions is also $\omega$-automatic: take the same \DPW{} as for $\R$, and swap the symbols of the pairs labeling its transitions. Hence, all the previous results remain true for player~$2$. Furthermore, a natural analogue of the classical inequality between upper and lower values~\cite{vonNeumannMorgenstern1953} holds in our relational setting: for any $\pi_1 \in \val_1(v)$ (witnessed by $\sigma_1$) and $\pi_2 \in \val_2(v)$ (witnessed by $\sigma_2$), the outcome $\rho = \outcomefrom{\sigma_1,\sigma_2}{v}$ satisfies both $\pi_1 \R \rho$ and $\rho \R \pi_2$. In particular, whenever $\R$ is transitive, we have $\pi_1 \R \pi_2$ for all $\pi_1 \in \val_1(v)$ and $\pi_2 \in \val_2(v)$.
\lipicsEnd
\end{remark}

\begin{remark}
    It is sometimes useful to study values and optimal strategies in subgames, i.e., when playing after some history (see, e.g.,~\cite{BrenguierRS14,BriceRB22,Gradel-Ummels-08}). In the context of this section, given a history $hv$ with $v \in V$, we define $\val_1(hv) = \{\pi \in \plays(v) \mid \exists \sigma_1 \in \Gamma_1, \forall \rho \in \plays_{\sigma_1}(v),\, h\pi \R h\rho\}$. Notice that the relation $\R^h$ defined by $(x,y) \in \mathop{\R^h}$ if and only if $(hx,hy) \in \mathop{\R}$, is also $\omega$-automatic. Indeed, it is accepted by the \DPW{} $\aut{A}$ for $\R$, where the initial state is replaced by the state reached after reading $(h,h)$. Thus, even if there are infinitely many histories, only $|\aut{A}|$ \DPWs{} are used to accept the various relations $\R^h$. The concept of optimal strategy is then easily adapted to subgames.
    \lipicsEnd
\end{remark}

\section{Nash Equilibria}
\label{section:nash-checking}

In this section, we come back to the general framework of multiplayer games $\game = (\arena, (\R_i)_{i\in\players})$, where for all $i\in\players$, $\R_i$ is an $\omega$-automatic preference relation for player~$i$. Let $\aut{A}_i$ be a \DPW{} with priority function $\alpha_i$, which accepts $\R_i$. We extend the threshold problem in this context based on Nash equilibria. A \emph{Nash Equilibrium} (NE) from an initial vertex $v$ is a strategy profile $\bsigma = (\sigma_i)_{i \in \players}$ such that for all players~$i$ and all strategies $\tau_i \in \Gamma_i$, we have $\outcomefrom{\bsigma}{v} \notR_i \outcomefrom{\tau_i,\bsigma_{-i}}{v}$, where $\bsigma_{-i}$ denotes $(\sigma_j)_{j \in \players \setminus \{i\}}$. So, NEs are strategy profiles where no single player has an incentive to unilaterally deviate from their strategy. When there exists a strategy $\tau_i \in \Gamma_i$ such that $\outcomefrom{\bsigma}{v} \R_i \outcomefrom{\tau_i,\bsigma_{-i}}{v}$, we say that $\tau_i$ (or, by abuse of notation, $\outcomefrom{\tau_i,\bsigma_{-i}}{v}$) is a \emph{profitable deviation} for player~$i$. We denote by $\NEout(v)$ the set of all NE outcomes starting at $v$.

The threshold problem in the context of multiplayer games is formulated as follows.

\begin{problem}[NE threshold problem] \label{problem}
    Given a game $\game = (\arena, (\R_i)_{i \in \players})$, an initial vertex $v$, and a lasso $\pi_i \in \plays(v)$ for each player~$i$, the \emph{NE threshold problem} is to decide whether there exists an NE $\bsigma$ from $v$ with outcome $\rho=\outcomefrom{\bsigma}{v}$ 
    such that $\pi_i \R_i \rho$ for all $i \in \players$.
\end{problem}

Before studying the NE threshold problem, we comment on the proposed NE definition.

\begin{remark}
    An alternative definition is proposed in~\cite{BouyerBMU15} for games with total preorders $\R_i$, $i \in \players$: a strategy profile $\bsigma$ is an NE if for all players $i$ and all strategies $\tau_i \in \Gamma_i$, we have $\outcomefrom{\tau_i,\bsigma_{-i}}{v} \R_i \outcomefrom{\bsigma}{v}$. We do not use this NE definition, as proposed in \cite{BruyereGrandmontRaskin25-mfcs} (see end of Section 6 in this paper). Recall that no hypotheses are imposed on our preference relations $\R_i$, $i \in \players$, except being $\omega$-automatic. Suppose now that some relation $\R_i$ used in $\game$ is reflexive. Then no strategy profile $\bsigma$ is an NE, because the deviating strategy $\tau_i = \sigma_i$ leads to $\outcomefrom{\bsigma}{v} \R_i \outcomefrom{\tau_i,\bsigma_{-i}}{v}$. Note that there are simple games with irreflexive preference relations that have no NEs, see \cite[Example 2]{BruyereGrandmontRaskin25-mfcs}. Our general framework is flexible and allows the use of (ir)reflexive relations~$\R_i$ or relations that are neither reflexive nor irreflexive.
    \lipicsEnd
\end{remark}

We have a result similar to the key result of \cref{section:zero-sum} (\cref{thm:values-automatic}): the set of NE outcomes starting at some vertex $v$ is an $\omega$-regular set (see \Cref{thm:apw-for-ne-outcome} below). This is the key lemma of this section from which we will derive several important corollaries. To obtain this lemma, 
we associate with $\game$ several two-player zero-sum games $\game_i = (\arena, \R_i)$, for each $i \in \players$, defined as follows. The arena of $\game_i$ is the same as in $\game$, where the two opposed players are player~$i$, and the set of players $\players \ssetminus \{i\}$, seen as one player~$-i$, called \emph{coalition}, and where $\R_i$ is the preference relation of player~$i$.

In an NE with outcome $\pi$, for each $i$, the strategy profile $\bsigma_{-i}$ must be both consistent with $\pi$ and ensure that player~$i$ has no incentive to deviate from $\pi$. In other words, the coalition $-i$ must follow $\pi$ while enforcing that no unilateral deviation by player~$i$ improves upon $\pi$. This is formalized by the following definition of \emph{consistent value} for player~$-i$:
\[
\cval_{-i}(v) = \{\pi \in \plays(v) \mid \exists \sigma_{-i} \in \Gamma_{-i},\; (\forall \rho \in \plays_{\sigma_{-i}}(v),\, \pi \notR_i \rho) \;\land\; \pi \in \plays_{\sigma_{-i}}(v)\}.
\]
Note that the last term requires that $\pi$ is consistent with $\sigma_{-i}$. The relation $\notR_i$ is $\omega$-automatic: it is accepted by the same \DPW{} $\aut{A}_i$ as for $\R_i$ (assuming it is complete), with its priorities increased by 1. We now show that $\NEout(v) = \bigcap_{i \in \players} \cval_{-i}(v)$ and each $\cval_{-i}(v)$ is accepted by a polynomial-size \APW{}, which produces a polynomial-size \APW{} for $\NEout(v)$.

\begin{lemma}[restate=lemapwnashequilibria,name=]
\label{thm:apw-for-ne-outcome}
    Let $\game$ be a game and $v_0 \in V$ be a vertex. The set $\NEout(v_0)$ is equal to $\bigcap_{i\in\players} \cval_{-i}(v_0)$, and is accepted by an \APW{} $\aut{B}$. Moreover,
    \begin{itemize}
        \item $\aut{B}$ has $\sum_{i \in \players} \bigO{|V|^2 \cdot |\aut{A}_i|}$ states and a priority function of index $\max_{i\in\players} (|\alpha_i|+1)$.
        \item If $\lang(\aut{B}) \ne \varnothing$, then there exists an NE whose strategies are all finite-memory.
    \end{itemize}
\end{lemma}

\begin{proof}
    First, we show that $\cval_{-i}(v_0)$ is accepted by an \APW{} $\aut{B}_i$ of polynomial size. We slightly modify the construction of \cref{thm:values-automatic} by considering $\notR_i$ and recording a Boolean flag $b \in \{\top, \bot\}$ in the state space, with $b = \top$ meaning that the current prefixes of $\pi$ and the simulated play $\rho$ are the same. The first transition goes from $s_{\text{init}}^i$ to $(v_0,v_0,\top,q_0^i)$. From $(v,v,b,q)$ with $b = \top$, if $v \in V_{-i}$, the coalition $-i$ is forced to follow $\pi$, if $v \in V_i$, player~$i$ universally branches between staying on $\pi$ ($b$ remains $\top$) and deviating to another successor of $v$ ($b$ becomes $\bot$). From $(v,u,b,q)$ with $b = \bot$, the transitions of the automaton are those described in \cref{thm:values-automatic} with $b$ remaining equal to $\bot$. The resulting \APW{} $\aut{B}_i$ has $\bigO{|V|^2\cdot |\aut{A}_i|}$ states and a priority function of index $|\alpha_i|+1$. Details are provided in \cref{app:apw-cval-construction}.

    Second, let us show that $\NEout(v_0) = \bigcap_{i \in \players} \cval_{-i}(v_0)$. Let $\pi = \outcomefrom{\bsigma}{v_0}$ be an NE outcome. Then, by definition, for all players~$i$, we have $\pi \in \plays_{\bsigma_{-i}}(v_0)$, and for all strategies $\tau_i \in \Gamma_i$, we have $\pi \notR_i \outcomefrom{\tau_i,\bsigma_{-i}}{v_0}$. Hence, in each game $\game_i$, this means that $\pi \in \cval_{-i}(v_0)$ (with $\bsigma_{-i}$ witnessing membership in $\cval_{-i}(v_0)$). It follows that $\pi \in \bigcap_{i \in \players} \cval_{-i}(v_0)$.\\
    \noindent Conversely, let $\pi$ be a play in $\bigcap_{i \in \players} \cval_{-i}(v_0)$. For each $i \in \players$, by definition of $\cval_{-i}(v_0)$, there exists a coalition strategy $\tau_{-i} \in \Gamma_{-i}$ such that $\pi \in \plays_{\tau_{-i}}(v_0)$ and for all $\rho \in \plays_{\tau_{-i}}(v_0)$, we have $\pi \notR_i \rho$. We construct a strategy profile $\bsigma = (\sigma_i)_{i \in \players}$ from $v_0$ as follows.
    Let $h$ be any prefix of $\pi$ ending in some vertex $v \in V_i$, and let $v'$ be the vertex following $h$ along $\pi$ (so that $hv'$ is a prefix of $\pi$). Then, from $h$, $\bsigma$ follows $\pi$, that is, $\sigma_i(h) = v'$. Let $u \neq v'$ be a successor of $v$ (so $hu$ is no longer a prefix of $\pi$, meaning that player~$i$ is the first player to deviate from $\pi$). Then, from $hu$,
    we define $\bsigma_{-i \| hu} = \tau_{-i \| hu}$ (splitting the strategy $\tau_{-i \| hu}$ among the players $j \in \players \setminus \{i\}$) and let $\sigma_{i \| hu}$ be an arbitrary strategy. Let us verify that $\bsigma$ is an NE with outcome $\pi$. By construction, $\outcomefrom{\bsigma}{v_0} = \pi$. Consider any player~$i$ and any strategy $\tau_i \in \Gamma_i$, and let $\rho = \outcomefrom{\tau_i, \bsigma_{-i}}{v_0}$. If $\rho = \pi$, then since $\pi \in \plays_{\tau_{-i}}(v_0)$, we immediately have $\pi \notR_i \pi$. Otherwise, let $hu$ be the shortest prefix of $\rho$ that is not a prefix of $\pi$. Since all players other than $i$ follow $\pi$ on prefixes of $\pi$, the last vertex of $h$ must belong to $V_i$. Because $\pi \in \plays_{\tau_{-i}}(v_0)$ and $h$ is a prefix of $\pi$ ending in $V_i$, the history $hu$ is consistent with $\tau_{-i}$. Since $\bsigma_{-i}$ plays according to $\tau_{-i \| hu}$ from $hu$ onwards, the entire play $\rho$ is consistent with $\tau_{-i}$, i.e., $\rho \in \plays_{\tau_{-i}}(v_0)$. By our choice of $\tau_{-i}$, this guarantees $\pi \notR_i \rho$. Thus, $\bsigma$ is an NE with outcome $\pi$.

    Third, we obtain the required \APW{} $\aut{B}$ accepting $\NEout(v_0) = \bigcap_{i \in \players} \lang(\aut{B}_i)$ as follows. We consider the \APWs{} $\aut{B}_i$ accepting $\cval_{-i}(v_0)$, and take their linear intersection~\cite{handbook-automata-kupferman18}. As each $\aut{B}_i$ has size $\bigO{|V|^2\cdot |\aut{A}_i|}$ and priority function of index $|\alpha_i|+1$, this leads to an \APW{} $\aut{B}$ with size $\sum_{i \in \players}\bigO{|V|^2\cdot |\aut{A}_i|}$ and a priority function of index $\max_{i\in\players} (|\alpha_i|+1)$.

    Finally, suppose that $\lang(\aut{B})$ is non-empty. Then $\aut{B}$ accepts a lasso $\pi$ of size exponential in the number of states and the index of $\aut{B}$~\cite{MiyanoHayashi1984,handbook-Automata-WilkeSchewe21}. For each $i \in \players$, membership of the lasso $\pi$ in $\lang(\aut{B}_i)$ reduces to a finite parity game of size $|S_i| \cdot |\pi|$ (as in the proof of \cref{thm:threshold-problem}), so a memoryless winning strategy in this parity game yields a finite-memory strategy $\tau_{-i}$ witnessing $\pi \in \cval_{-i}(v_0)$. Combining the lasso $\pi$ with these finite-memory punishment strategies $\tau_{-i}$ as above yields an NE composed entirely of finite-memory strategies.
\end{proof}

We derive several important corollaries of \cref{thm:apw-for-ne-outcome}. We settle the complexity class of the NE threshold problem, but also of a simpler variant: whether there exists an NE. In both cases, we close the complexity gap left open in~\cite{BruyereGrandmontRaskin25-mfcs} (Theorems~5 and~6 where a \twoExptime{}-membership was established).

\begin{theorem}[restate=thmnethreshold,name=]
\label{thm:NEthreshold}
    The NE threshold problem and the NE existence problem are \pspaceComplete{}.
\end{theorem}

\begin{proof}
    For both problems, the \pspaceHard{}ness is proved in~\cite{BruyereGrandmontRaskin25-mfcs}. Let us focus on the \pspace{}-membership. We begin with the NE existence problem. Keeping the notations of \cref{thm:apw-for-ne-outcome}, solving this problem is equivalent to checking whether $\lang(\aut{B}) \neq \varnothing$ for $\aut{B}$ an \APW{} of polynomial size, whose non-emptiness can be checked in polynomial space, as done in the proof of \cref{thm:threshold-problem}.
    Let us now solve the NE threshold problem: given a lasso $\pi_i \in \plays(v)$ for all players~$i$, does there exist an NE outcome $\rho$ from $v$ such that $\pi_i \R_i \rho$, for all $i$? With the product of $\aut{A}_i$ with $\pi_i$, we get a \DPW{} $\aut{A}_i'$ accepting the language $L_i = \{\rho \in V^\omega \mid \pi_i \R_i \rho\}$ of polynomial size $|\aut{A}_i|\cdot |\pi_i|$. Then, there exists a solution to the threshold problem instance if and only if $L = \NEout(v) \cap (\bigcap_{i \in \players} L_i) \neq \varnothing$. As done in the proof of \cref{thm:apw-for-ne-outcome}, from $\aut{B}$ and each $\aut{A}_i'$, we finally construct an \APW{} of polynomial size that accepts $L$. Checking the non-emptiness of $L$ is therefore in polynomial space.
\end{proof}

Let us mention that two verification problems related to the NE existence problem were solved in~\cite{BruyereGrandmontRaskin25-mfcs}. First, deciding whether a strategy profile $\bsigma = (\sigma_i)_{i\in \players}$, where each strategy $\sigma_i$ is defined by a Mealy machine, is an NE, is a \pspaceComplete{} problem~\cite[Theorem~3]{BruyereGrandmontRaskin25-mfcs}. Second, deciding whether a lasso is an NE outcome is a \parityComplete{} problem~\cite[Theorem~4]{BruyereGrandmontRaskin25-mfcs}.

Using \cref{thm:apw-for-ne-outcome}, we can handle another variant of the NE threshold problem, as studied for classical settings in~\cite{AlmagorKupfermanPerelli2018,Ummels08}. In this variant, we ask whether there exists an NE outcome that satisfies an $\omega$-regular condition or an LTL formula. In both cases, we get \pspace{}-membership.

\begin{corollary}
\label{thm:LTLconstraint}
    Let $\game$ be a game, $v \in V$ be a vertex, and $C\subseteq V^\omega$ be a language (called constraint), given either as an \APW{} or an LTL formula. Deciding whether there exists an NE from $v$ whose outcome belongs to $C$ is \pspaceComplete{}.
\end{corollary}

\begin{proof}
    We have to decide whether $\NEout(v)\cap C\neq \varnothing$. We can assume w.l.o.g.\ that the constraint $C$ is given as an \APW{}, because any LTL formula can be transformed into an \ABW{} of linear size~\cite{Vardi95} (seen as an \APW{}). Therefore, as in \cref{thm:NEthreshold}, the \pspace{} algorithm amounts to running a non-emptiness check on an \APW{}. 
    Since the NE existence problem is \pspaceComplete{} (with constraint $C = V^\omega$) by \cref{thm:NEthreshold}, we immediately get the \pspace{}-hardness.
\end{proof}

\begin{remark}
    When there exist several NEs, it is interesting to focus on the NEs with the ``most preferred'' outcomes, that is, on \emph{Pareto-optimal} NE outcomes. Intuitively, an NE outcome $\pi$ is Pareto-optimal if no player~$i$ can improve in another NE outcome $\pi'$ ($\pi \R_i \pi'$) without making at least one player~$j$ worse off ($\pi' \R_j \pi$). Formally:
    \[
    \PO(\NEout(v)) = \left\{\pi \in \NEout(v) \mid \forall \pi' \in \NEout(v), (\exists i \in \players,\, \pi \R_i \pi') \Rightarrow (\exists j \in \players, \pi' \R_j \pi)\right\}
    \]
    Note that when each $\R_i$ is a strict weak order, this definition coincides with the classical notion of Pareto optimality.
    By using again the \APW{} $\aut{B}$ accepting $\NEout(v)$ (\cref{thm:apw-for-ne-outcome}), we can construct an automaton for $\PO(\NEout(v))$ and decide the existence of a Pareto-optimal NE in exponential space. The arguments are as follows.
    We first consider the language 
    \[L_1 = 
    \left\{(\pi,\pi') \in V^\omega \times V^\omega \mid \pi' \not\in \NEout(v) \lor \big((\exists i \in \players, \pi \R_i \pi') \Rightarrow (\exists j \in \players, \pi' \R_j \pi)\big)\right\}.
    \]
    From $\aut{B}$ and the \DPWs{} $\aut{A}_i$ (seen as \APWs{}) accepting the relations $\R_i$, $i \in \players$, as \APWs{} are closed under Boolean operations, we can construct an \APW{} $\aut{C}_1$ accepting $L_1$. We then transform $\aut{C}_1$ into an \NBW{} of exponential size accepting the complementary language $(V \times V)^\omega \setminus L_1$~\cite{DaxKlaedtke08}, and then take its projection on the $\pi$-component. The resulting \NBW{} $\aut{C}_2$ accepts the language
    \[L_2 = 
    \left\{\pi \in V^\omega \mid \exists \pi' \in \NEout(v), (\exists i \in \players, \pi \R_i \pi') \land (\forall j \in \players, \pi' \notR_j \pi)\right\}.
    \]
    Notice that $\PO(\NEout(v))$ is equal to $\NEout(v) \cap (V^\omega \setminus L_2)$. Therefore, from the \APW{} $\aut{B}$ accepting $\NEout(v)$ and the \NBW{} $\aut{C}_2$ seen as an \APW{}, we can construct an \APW{} $\aut{C}$ accepting $\PO(\NEout(v))$. This automaton is of exponential size, because $\aut{C}_2$ is. Finally, we perform a non-emptiness check on $\aut{C}$ in polynomial space (in the size of $\aut{C}$), and thus in exponential space overall.
    
    Note that \pspace{}-hardness directly follows from the \pspace{}-hardness of the NE existence problem for one-player games~\cite{BruyereGrandmontRaskin25-mfcs}. Indeed, for such games, any $\pi \in \NEout(v)$ satisfies $\pi \notR_1 \pi'$ for all $\pi' \in \plays(v)$, so $\PO(\NEout(v)) = \NEout(v)$ and thus $\NEout(v) \neq \varnothing$ if and only if $\PO(\NEout(v)) \neq \varnothing$.

    Among all NEs, we could also be interested in those with the most preferred outcomes limited to some particular players. The definition of $\PO(\NEout(v))$ can be easily adapted to a strict subset of players (instead of the whole set $\players$). The above reasoning still holds, even for a restriction to one player.
    \lipicsEnd
\end{remark}

\section{Rational Synthesis}
\label{section:rational-synthesis}

Within the context of rational synthesis~\cite{Bruyere21,ConduracheFGR16,FismanKL10,KupfermanPV16}, we consider a game $\game$ with $\players = \{0,1, \dots, t\}$ such that player~$0$ is a specific player, often called \emph{system} or \emph{leader}, and the other players~$1,\dots,t$ compose the \emph{environment} and are called \emph{followers}. Player~$0$ announces their strategy $\sigma_0$ at the beginning of the game and is not allowed to change it according to the behavior of the other players. The response of those players to $\sigma_0$ is assumed to be \emph{rational}, where the rationality can be described as an NE outcome in the following way. Given an initial vertex $v$ and a strategy $\sigma_0$ announced by player~$0$, a strategy profile $\bsigma = (\sigma_0,\bsigma_{-0})$ is called a \emph{$0$-fixed Nash equilibrium} from $v$ if for every player~$i\in\players\ssetminus\{0\}$ and every strategy $\tau_i \in \Gamma_i$, it holds that $\outcomefrom{\bsigma}{v} \notR_i \outcomefrom{\tau_i,\bsigma_{-i}}{v}$.
We also say that $\bsigma$ is a \emph{$\sigma_0$-fixed NE} to emphasize the strategy $\sigma_0$ of player~$0$. The strategy profile $\bsigma_{-0}$ is the \emph{NE (rational) response} of players~$1, \ldots, t$ to the strategy $\sigma_0$ announced by player~$0$.

We consider the problem of \emph{NE rational synthesis}. Given a lasso $\pi$, the goal is to synthesize a strategy $\sigma_0$ for player~$0$ that guarantees that $\pi \R_0 {\outcomefrom{\bsigma}{v}}$ from $v$ whatever the NE response $\bsigma_{-0}$ of the environment is. We also consider the simpler problem where the environment \emph{cooperates} with the leader by proposing \emph{some} NE response $\bsigma_{-0}$ that guarantees that $\pi \R_0 {\outcomefrom{\bsigma}{v}}$. Both problems can be seen as a kind of threshold problem, where the threshold constraint $\pi$ must be enforced by player~$0$ against an environment that can be cooperative or not. 

\begin{problem}[NE rational synthesis]
\label{problem:synthesis}
    Given a game $\game = (\arena,(\R_i)_{i \in \players})$, an initial vertex $v$, and a lasso $\pi \in \plays(v)$,
    \begin{itemize}
        \item the \emph{cooperative NE rational synthesis (CRS)} problem is to decide whether there exists $\sigma_0 \in \Gamma_0$ and a $\sigma_0$-fixed NE $\bsigma$ from $v$ such that $\pi \R_0 {\outcomefrom{\bsigma}{v}}$;
        \item the \emph{non-cooperative NE rational synthesis (NCRS)} problem is to decide whether there exists $\sigma_0 \in \Gamma_0$ such that for all $\sigma_0$-fixed NEs $\bsigma$ from $v$, it holds that $\pi \R_0 {\outcomefrom{\bsigma}{v}}$.
    \end{itemize}
\end{problem} 

Our main results are that the CRS problem has the same complexity class as the NE threshold problem, while it becomes undecidable in the non-cooperative case, $\Sigma_1^1$-hard~\cite{Hinman17}, hence beyond the arithmetical hierarchy.

\begin{theorem}[restate=thmrs,name=]
\label{thm:crs-pspace-ncrs-undec}
    \begin{itemize}
        \item The cooperative NE rational synthesis problem is \pspaceComplete{}.
        \item The non-cooperative NE rational synthesis problem is undecidable ($\Sigma_1^1$-hard), already for two players.
    \end{itemize}
\end{theorem}

\begin{proof}[Proof (Sketch)]
    We here prove the non-cooperative case; the cooperative case is similar to \cref{thm:NEthreshold} and detailed in \cref{app:rational-synthesis}. The proof is inspired by an undecidability result for tree automata with subtree equality constraints~\cite{tata}, and relies on a reduction from the \emph{infinite Post correspondence problem} ($\omega$-PCP): given two alphabets $A, B$ and morphisms $\varphi_L,\varphi_R: A \to B^+$, the problem is to decide whether there exists an infinite word $u \in A^\omega$ such that $\varphi_L(u) = \varphi_R(u)$. We construct a two-player game $\game = (\mathsf{Arn}, (\R_0,\R_1))$ and a lasso $\pi$, such that the $\omega$-PCP instance has a solution if and only if the NCRS instance with threshold $\pi$ has a solution. This reduction gives undecidability; we further show in \cref{app:rational-synthesis} that a slight modification of the proof yields $\Sigma_1^1$-hardness.

\begin{figure}[t]
    \centering
    \begin{minipage}[c]{0.43\textwidth}
        \centering
        \begin{tikzpicture}[x=0.75pt,y=0.75pt,yscale=-1,scale=.8,every node/.style={scale=.8}]
            \draw [fill={rgb,255:red,155;green,155;blue,155},fill opacity=0.1,dash pattern={on 4.5pt off 4.5pt},line width=0.75] (353.45,118.03) .. controls (353.45,112.15) and (358.22,107.39) .. (364.09,107.39) -- (396.01,107.39) .. controls (401.88,107.39) and (406.65,112.15) .. (406.65,118.03) -- (406.65,159.8) .. controls (406.65,165.68) and (401.88,170.44) .. (396.01,170.44) -- (364.09,170.44) .. controls (358.22,170.44) and (353.45,165.68) .. (353.45,159.8) -- cycle;
            \draw (148.82,139.17) .. controls (148.82,133.5) and (153.42,128.9) .. (159.09,128.9) .. controls (164.76,128.9) and (169.36,133.5) .. (169.36,139.17) .. controls (169.36,144.85) and (164.76,149.45) .. (159.09,149.45) .. controls (153.42,149.45) and (148.82,144.85) .. (148.82,139.17) -- cycle;
            \draw (209.6,130.57) .. controls (209.6,130.57) and (209.6,130.57) .. (209.6,130.57) -- (226.5,130.57) .. controls (226.5,130.57) and (226.5,130.57) .. (226.5,130.57) -- (226.5,148.06) .. controls (226.5,148.06) and (226.5,148.06) .. (226.5,148.06) -- (209.6,148.06) .. controls (209.6,148.06) and (209.6,148.06) .. (209.6,148.06) -- cycle;
            \draw (164.56,130.37) -- (183.71,105.9);
            \draw [shift={(185.56,103.53)}, rotate = 128.04,fill={rgb,255:red,0;green,0;blue,0},line width=0.08,draw opacity=0] (5.36,-2.57) -- (0,0) -- (5.36,2.57) -- (3.56,0) -- cycle;
            \draw (164.56,147.97) -- (184.08,173.48);
            \draw [shift={(185.9,175.86)}, rotate = 232.58,fill={rgb,255:red,0;green,0;blue,0},line width=0.08,draw opacity=0] (5.36,-2.57) -- (0,0) -- (5.36,2.57) -- (3.56,0) -- cycle;
            \draw (168.89,136.76) .. controls (180.25,131.36) and (196.78,132.24) .. (206.87,135.8);
            \draw [shift={(209.5,136.86)}, rotate = 204.62,fill={rgb,255:red,0;green,0;blue,0},line width=0.08,draw opacity=0] (5.36,-2.57) -- (0,0) -- (5.36,2.57) -- (3.56,0) -- cycle; 
            \draw (209.3,143.66) .. controls (201.25,148.42) and (181.31,148.49) .. (170.93,144.79);
            \draw [shift={(168.3,143.66)}, rotate = 27.55,fill={rgb,255:red,0;green,0;blue,0},line width=0.08,draw opacity=0] (5.36,-2.57) -- (0,0) -- (5.36,2.57) -- (3.56,0) -- cycle; 
            \draw (226.96,143.97) -- (251.67,157.43);
            \draw [shift={(254.3,158.86)}, rotate = 208.57,fill={rgb,255:red,0;green,0;blue,0},line width=0.08,draw opacity=0] (5.36,-2.57) -- (0,0) -- (5.36,2.57) -- (3.56,0) -- cycle;
            \draw (226.56,135.17) -- (250,120.89);
            \draw [shift={(252.56,119.33)}, rotate = 148.65,fill={rgb,255:red,0;green,0;blue,0},line width=0.08,draw opacity=0] (5.36,-2.57) -- (0,0) -- (5.36,2.57) -- (3.56,0) -- cycle;
            \draw (252.56,117.33) .. controls (252.56,111.66) and (257.16,107.06) .. (262.83,107.06) .. controls (268.51,107.06) and (273.11,111.66) .. (273.11,117.33) .. controls (273.11,123.01) and (268.51,127.61) .. (262.83,127.61) .. controls (257.16,127.61) and (252.56,123.01) .. (252.56,117.33) -- cycle;
            \draw (254.3,160.86) .. controls (254.3,155.19) and (258.9,150.59) .. (264.57,150.59) .. controls (270.25,150.59) and (274.85,155.19) .. (274.85,160.86) .. controls (274.85,166.53) and (270.25,171.13) .. (264.57,171.13) .. controls (258.9,171.13) and (254.3,166.53) .. (254.3,160.86) -- cycle;
            \draw (135.82,139.33) -- (145.82,139.21);
            \draw [shift={(148.82,139.17)}, rotate = 179.29,fill={rgb,255:red,0;green,0;blue,0},line width=0.08,draw opacity=0] (5.36,-2.57) -- (0,0) -- (5.36,2.57) -- (3.56,0) -- cycle;
            \draw (273.11,117.33) -- (300.68,117.27);
            \draw [shift={(303.68,117.26)}, rotate = 179.86,fill={rgb,255:red,0;green,0;blue,0},line width=0.08,draw opacity=0] (5.36,-2.57) -- (0,0) -- (5.36,2.57) -- (3.56,0) -- cycle;
            \draw (320.48,117.26) -- (349.05,117.19);
            \draw [shift={(352.05,117.19)}, rotate = 179.87,fill={rgb,255:red,0;green,0;blue,0},line width=0.08,draw opacity=0] (5.36,-2.57) -- (0,0) -- (5.36,2.57) -- (3.56,0) -- cycle;
            \draw (275.11,161.33) -- (302.68,161.27);
            \draw [shift={(305.68,161.26)}, rotate = 179.86,fill={rgb,255:red,0;green,0;blue,0},line width=0.08,draw opacity=0] (5.36,-2.57) -- (0,0) -- (5.36,2.57) -- (3.56,0) -- cycle;
            \draw (320.48,161.26) -- (350.05,161.19);
            \draw [shift={(353.05,161.19)}, rotate = 179.87,fill={rgb,255:red,0;green,0;blue,0},line width=0.08,draw opacity=0] (5.36,-2.57) -- (0,0) -- (5.36,2.57) -- (3.56,0) -- cycle;
    
            \draw (373,134) node [anchor=north west,inner sep=0.75pt,align=left] {$B$};
            \draw (213,135.93) node [anchor=north west,inner sep=0.75pt,align=left] {$a$};
            \draw (253.8,110.37) node [anchor=north west,inner sep=0.75pt,align=left] {$L_{a}$};
            \draw (254.6,154.17) node [anchor=north west,inner sep=0.75pt,align=left] {$R_{a}$};
            \draw (187.13,101.33) node [anchor=north west,inner sep=0.75pt,font=\scriptsize,align=left] {$...$};
            \draw (188.13,176.53) node [anchor=north west,inner sep=0.75pt,font=\scriptsize,align=left] {$...$};
            \draw (305.73,115.33) node [anchor=north west,inner sep=0.75pt,font=\scriptsize,align=left] {$...$};
            \draw (278.13,97.33) node [anchor=north west,inner sep=0.75pt,font=\small,align=left] {Read $\varphi_{L}(a)$};
            \draw (306.73,160.33) node [anchor=north west,inner sep=0.75pt,font=\scriptsize,align=left] {$...$};
            \draw (280.13,141.33) node [anchor=north west,inner sep=0.75pt,font=\small,align=left] {Read $\varphi_{R}(a)$};
            \draw (151.75,132) node [anchor=north west,inner sep=0.75pt,align=left] {$\#$};
        \end{tikzpicture}
        \caption{The arena $\mathsf{Arn}$ used in the undecidability result of \cref{thm:crs-pspace-ncrs-undec}.}
        \label{fig:arena-undecidability-ncrs}
    \end{minipage}
    \hfill
    \begin{minipage}[c]{0.55\textwidth}
        \centering
        \begin{tikzpicture}[x=0.75pt,y=0.75pt,yscale=-1,scale=.8,every node/.style={scale=.8}]
            \draw (301.96,30.1) .. controls (307.64,30.1) and (312.23,34.7) .. (312.23,40.37) .. controls (312.23,46.05) and (307.63,50.65) .. (301.96,50.64) .. controls (296.28,50.64) and (291.69,46.04) .. (291.69,40.37) .. controls (291.69,34.69) and (296.29,30.1) .. (301.96,30.1) -- cycle;
            \draw (310.55,61.89) .. controls (310.55,61.89) and (310.55,61.89) .. (310.55,61.89) -- (310.54,78.79) .. controls (310.54,78.79) and (310.54,78.79) .. (310.54,78.79) -- (293.05,78.78) .. controls (293.05,78.78) and (293.05,78.78) .. (293.05,78.78) -- (293.06,61.88) .. controls (293.06,61.88) and (293.06,61.88) .. (293.06,61.88) -- cycle;
            \draw (310.01,69.89) -- (374.02,92.89);
            \draw [shift={(376.84,93.9)}, rotate = 199.76,fill={rgb,255:red,0;green,0;blue,0},line width=0.08,draw opacity=0] (5.36,-2.57) -- (0,0) -- (5.36,2.57) -- (3.56,0) -- cycle;
            \draw (292.49,70.51) -- (226.22,91.39);
            \draw [shift={(223.36,92.29)}, rotate = 342.51,fill={rgb,255:red,0;green,0;blue,0},line width=0.08,draw opacity=0] (5.36,-2.57) -- (0,0) -- (5.36,2.57) -- (3.56,0) -- cycle;
            \draw (212.76,85.99) .. controls (219.2,85.99) and (224.42,91.22) .. (224.42,97.67) .. controls (224.42,104.11) and (219.19,109.34) .. (212.75,109.33) .. controls (206.3,109.33) and (201.08,104.1) .. (201.08,97.66) .. controls (201.08,91.21) and (206.31,85.99) .. (212.76,85.99) -- cycle;
            \draw (301.96,50.64) -- (302,58.22);
            \draw [shift={(302.02,61.22)}, rotate = 269.66,fill={rgb,255:red,0;green,0;blue,0},line width=0.08,draw opacity=0] (5.36,-2.57) -- (0,0) -- (5.36,2.57) -- (3.56,0) -- cycle;
            \draw (301.96,79.14) -- (302,88.72);
            \draw [shift={(302.02,91.72)}, rotate = 269.71,fill={rgb,255:red,0;green,0;blue,0},line width=0.08,draw opacity=0] (5.36,-2.57) -- (0,0) -- (5.36,2.57) -- (3.56,0) -- cycle;
            \draw (193,135.58) .. controls (193,130.93) and (196.77,127.16) .. (201.42,127.16) -- (223.29,127.16) .. controls (227.94,127.16) and (231.71,130.93) .. (231.71,135.58) -- (231.71,135.58) .. controls (231.71,140.23) and (227.94,144) .. (223.29,144) -- (201.42,144) .. controls (196.77,144) and (193,140.23) .. (193,135.58) -- cycle;
            \draw (212.35,108.83) -- (212.4,124.41);
            \draw [shift={(212.41,127.41)}, rotate = 269.81,fill={rgb,255:red,0;green,0;blue,0},line width=0.08,draw opacity=0] (5.36,-2.57) -- (0,0) -- (5.36,2.57) -- (3.56,0) -- cycle;
            \draw (211.85,143.83) -- (211.9,159.41);
            \draw [shift={(211.91,162.41)}, rotate = 269.81,fill={rgb,255:red,0;green,0;blue,0},line width=0.08,draw opacity=0] (5.36,-2.57) -- (0,0) -- (5.36,2.57) -- (3.56,0) -- cycle;
            \draw (387.35,111.83) -- (387.4,127.41);
            \draw [shift={(387.41,130.41)}, rotate = 269.81,fill={rgb,255:red,0;green,0;blue,0},line width=0.08,draw opacity=0] (5.36,-2.57) -- (0,0) -- (5.36,2.57) -- (3.56,0) -- cycle;
            \draw (386.85,146.83) -- (386.9,162.41);
            \draw [shift={(386.91,165.41)}, rotate = 269.81,fill={rgb,255:red,0;green,0;blue,0},line width=0.08,draw opacity=0] (5.36,-2.57) -- (0,0) -- (5.36,2.57) -- (3.56,0) -- cycle;
            \draw (301.96,92) .. controls (307.64,92) and (312.23,96.6) .. (312.23,102.27) .. controls (312.23,107.95) and (307.63,112.55) .. (301.96,112.54) .. controls (296.28,112.54) and (291.69,107.94) .. (291.69,102.27) .. controls (291.69,96.59) and (296.29,92) .. (301.96,92) -- cycle;
            \draw (310.55,125.79) .. controls (310.55,125.79) and (310.55,125.79) .. (310.55,125.79) -- (310.54,142.69) .. controls (310.54,142.69) and (310.54,142.69) .. (310.54,142.69) -- (293.05,142.68) .. controls (293.05,142.68) and (293.05,142.68) .. (293.05,142.68) -- (293.06,125.78) .. controls (293.06,125.78) and (293.06,125.78) .. (293.06,125.78) -- cycle;
            \draw (311.01,133.79) -- (337.14,146.49);
            \draw [shift={(339.84,147.8)}, rotate = 205.91,fill={rgb,255:red,0;green,0;blue,0},line width=0.08,draw opacity=0] (5.36,-2.57) -- (0,0) -- (5.36,2.57) -- (3.56,0) -- cycle;
            \draw (292.49,134.41) -- (269.01,146.79);
            \draw [shift={(266.36,148.19)}, rotate = 332.2,fill={rgb,255:red,0;green,0;blue,0},line width=0.08,draw opacity=0] (5.36,-2.57) -- (0,0) -- (5.36,2.57) -- (3.56,0) -- cycle;
            \draw (301.96,112.54) -- (302,122.12);
            \draw [shift={(302.02,125.12)}, rotate = 269.71,fill={rgb,255:red,0;green,0;blue,0},line width=0.08,draw opacity=0] (5.36,-2.57) -- (0,0) -- (5.36,2.57) -- (3.56,0) -- cycle;
            \draw (301.96,143.04) -- (302,152.62);
            \draw [shift={(302.02,155.62)}, rotate = 269.71,fill={rgb,255:red,0;green,0;blue,0},line width=0.08,draw opacity=0] (5.36,-2.57) -- (0,0) -- (5.36,2.57) -- (3.56,0) -- cycle;
            \draw (301.46,156) .. controls (307.14,156) and (311.73,160.6) .. (311.73,166.27) .. controls (311.73,171.95) and (307.13,176.55) .. (301.46,176.54) .. controls (295.78,176.54) and (291.19,171.94) .. (291.19,166.27) .. controls (291.19,160.59) and (295.79,156) .. (301.46,156) -- cycle;
            \draw (310.05,189.79) .. controls (310.05,189.79) and (310.05,189.79) .. (310.05,189.79) -- (310.04,206.69) .. controls (310.04,206.69) and (310.04,206.69) .. (310.04,206.69) -- (292.55,206.68) .. controls (292.55,206.68) and (292.55,206.68) .. (292.55,206.68) -- (292.56,189.78) .. controls (292.56,189.78) and (292.56,189.78) .. (292.56,189.78) -- cycle;
            \draw (301.46,176.54) -- (301.5,186.12);
            \draw [shift={(301.52,189.12)}, rotate = 269.71,fill={rgb,255:red,0;green,0;blue,0},line width=0.08,draw opacity=0] (5.36,-2.57) -- (0,0) -- (5.36,2.57) -- (3.56,0) -- cycle;
            \draw (301.46,207.04) -- (301.51,218.62);
            \draw [shift={(301.52,221.62)}, rotate = 269.75,fill={rgb,255:red,0;green,0;blue,0},line width=0.08,draw opacity=0] (5.36,-2.57) -- (0,0) -- (5.36,2.57) -- (3.56,0) -- cycle;
            \draw (310.01,202.79) -- (319.71,219.04);
            \draw [shift={(321.25,221.62)}, rotate = 239.15,fill={rgb,255:red,0;green,0;blue,0},line width=0.08,draw opacity=0] (5.36,-2.57) -- (0,0) -- (5.36,2.57) -- (3.56,0) -- cycle;
            \draw (292.75,202.62) -- (282.38,218.6);
            \draw [shift={(280.75,221.12)}, rotate = 302.97,fill={rgb,255:red,0;green,0;blue,0},line width=0.08,draw opacity=0] (5.36,-2.57) -- (0,0) -- (5.36,2.57) -- (3.56,0) -- cycle;
            \draw (387.26,87.99) .. controls (393.7,87.99) and (398.92,93.22) .. (398.92,99.67) .. controls (398.92,106.11) and (393.69,111.34) .. (387.25,111.33) .. controls (380.8,111.33) and (375.58,106.1) .. (375.58,99.66) .. controls (375.58,93.21) and (380.81,87.99) .. (387.26,87.99) -- cycle;
            \draw (368,138.08) .. controls (368,133.43) and (371.77,129.66) .. (376.42,129.66) -- (398.29,129.66) .. controls (402.94,129.66) and (406.71,133.43) .. (406.71,138.08) -- (406.71,138.08) .. controls (406.71,142.73) and (402.94,146.5) .. (398.29,146.5) -- (376.42,146.5) .. controls (371.77,146.5) and (368,142.73) .. (368,138.08) -- cycle;
            \draw (258.76,145.49) .. controls (265.2,145.49) and (270.42,150.72) .. (270.42,157.17) .. controls (270.42,163.61) and (265.19,168.84) .. (258.75,168.83) .. controls (252.3,168.83) and (247.08,163.6) .. (247.08,157.16) .. controls (247.08,150.71) and (252.31,145.49) .. (258.76,145.49) -- cycle;
            \draw (239,195.08) .. controls (239,190.43) and (242.77,186.66) .. (247.42,186.66) -- (269.29,186.66) .. controls (273.94,186.66) and (277.71,190.43) .. (277.71,195.08) -- (277.71,195.08) .. controls (277.71,199.73) and (273.94,203.5) .. (269.29,203.5) -- (247.42,203.5) .. controls (242.77,203.5) and (239,199.73) .. (239,195.08) -- cycle;
            \draw (258.35,168.33) -- (258.4,183.91);
            \draw [shift={(258.41,186.91)}, rotate = 269.81,fill={rgb,255:red,0;green,0;blue,0},line width=0.08,draw opacity=0] (5.36,-2.57) -- (0,0) -- (5.36,2.57) -- (3.56,0) -- cycle;
            \draw (257.85,203.33) -- (257.9,218.91);
            \draw [shift={(257.91,221.91)}, rotate = 269.81,fill={rgb,255:red,0;green,0;blue,0},line width=0.08,draw opacity=0] (5.36,-2.57) -- (0,0) -- (5.36,2.57) -- (3.56,0) -- cycle;
            \draw (347.85,168.33) -- (347.9,183.91);
            \draw [shift={(347.91,186.91)}, rotate = 269.81,fill={rgb,255:red,0;green,0;blue,0},line width=0.08,draw opacity=0] (5.36,-2.57) -- (0,0) -- (5.36,2.57) -- (3.56,0) -- cycle;
            \draw (347.35,203.33) -- (347.4,218.91);
            \draw [shift={(347.41,221.91)}, rotate = 269.81,fill={rgb,255:red,0;green,0;blue,0},line width=0.08,draw opacity=0] (5.36,-2.57) -- (0,0) -- (5.36,2.57) -- (3.56,0) -- cycle;
            \draw (348.76,144.49) .. controls (355.2,144.49) and (360.42,149.72) .. (360.42,156.17) .. controls (360.42,162.61) and (355.19,167.84) .. (348.75,167.83) .. controls (342.3,167.83) and (337.08,162.6) .. (337.08,156.16) .. controls (337.08,149.71) and (342.31,144.49) .. (348.76,144.49) -- cycle;
            \draw (328.5,194.58) .. controls (328.5,189.93) and (332.27,186.16) .. (336.92,186.16) -- (358.79,186.16) .. controls (363.44,186.16) and (367.21,189.93) .. (367.21,194.58) -- (367.21,194.58) .. controls (367.21,199.23) and (363.44,203) .. (358.79,203) -- (336.92,203) .. controls (332.27,203) and (328.5,199.23) .. (328.5,194.58) -- cycle;
            \draw [color={rgb,255:red,255;green,0;blue,0},draw opacity=0.5 ,line width=1.5,dash pattern={on 5.63pt off 4.5pt}]  (228.61,213.32) .. controls (179.07,208.82) and (145.68,159.62) .. (171.86,142.7);
            \draw [shift={(174.94,140.99)}, rotate = 154.65,fill={rgb,255:red,255;green,0;blue,0},fill opacity=0.5 ,line width=0.08,draw opacity=0] (8.75,-4.2) -- (0,0) -- (8.75,4.2) -- (5.81,0) -- cycle;
            \draw [color={rgb,255:red,255;green,0;blue,0},draw opacity=0.5 ,line width=1.5,dash pattern={on 5.63pt off 4.5pt}]  (420.28,146.66) .. controls (460.66,151.91) and (460.95,198.88) .. (394.37,207.62);
            \draw [shift={(391.28,207.99)}, rotate = 353.72,fill={rgb,255:red,255;green,0;blue,0},fill opacity=0.5 ,line width=0.08,draw opacity=0] (8.75,-4.2) -- (0,0) -- (8.75,4.2) -- (5.81,0) -- cycle;
            \draw [color={rgb,255:red,255;green,0;blue,0},draw opacity=0.5 ,line width=1.5,dash pattern={on 5.63pt off 4.5pt}]  (225.61,83.66) .. controls (255.51,47.91) and (341.5,47.65) .. (374.51,81.64);
            \draw [shift={(376.94,84.32)}, rotate = 229.79,fill={rgb,255:red,255;green,0;blue,0},fill opacity=0.5 ,line width=0.08,draw opacity=0] (8.75,-4.2) -- (0,0) -- (8.75,4.2) -- (5.81,0) -- cycle;
            \draw (294.67,33.43) node [anchor=north west,inner sep=0.75pt,align=left] {$\#$};
            \draw (201.97,89.7) node [anchor=north west,inner sep=0.75pt,align=left] {$L_{a_{1}}$};
            \draw (206.47,168.87) node [anchor=north west,inner sep=0.75pt,font=\scriptsize,align=left] {$...$};
            \draw (295.02,66.1) node [anchor=north west,inner sep=0.75pt,align=left] {$a_{1}$};
            \draw (374.77,92.33) node [anchor=north west,inner sep=0.75pt,align=left] {$R_{a_{1}}$};
            \draw (193.47,129.2) node [anchor=north west,inner sep=0.75pt,font=\footnotesize,align=left] {$\varphi_{L}(a_{1})$};
            \draw (186.97,178.7) node [anchor=north west,inner sep=0.75pt,font=\footnotesize,align=left] {$u_{1}\in B^\omega$};
            \draw (381.47,172.87) node [anchor=north west,inner sep=0.75pt,font=\scriptsize,align=left] {$...$};
            \draw (370.97,181.7) node [anchor=north west,inner sep=0.75pt,font=\footnotesize,align=left] {$v_{1}\in B^\omega$};
            \draw (251.47,224.77) node [anchor=north west,inner sep=0.75pt,font=\scriptsize,align=left] {$...$};
            \draw (294,130.9) node [anchor=north west,inner sep=0.75pt,align=left] {$a_{2}$};
            \draw (234.3,231.6) node [anchor=north west,inner sep=0.75pt,font=\footnotesize,align=left] {$u_{2}\in B^\omega$};
            \draw (326.63,230.6) node [anchor=north west,inner sep=0.75pt,font=\footnotesize,align=left] {$v_{2}\in B^\omega$};
            \draw (294.2,194.2) node [anchor=north west,inner sep=0.75pt,align=left] {$a_{3}$};
            \draw (368.47,132.2) node [anchor=north west,inner sep=0.75pt,font=\footnotesize,align=left] {$\varphi_{R}(a_{1})$};
            \draw (247.97,149.2) node [anchor=north west,inner sep=0.75pt,align=left] {$L_{a_2}$};
            \draw (239.47,188.7) node [anchor=north west,inner sep=0.75pt,font=\footnotesize,align=left] {$\varphi_{L}( a_{2})$};
            \draw (336.27,148.83) node [anchor=north west,inner sep=0.75pt,align=left] {$R_{a_{2}}$};
            \draw (340.97,224.87) node [anchor=north west,inner sep=0.75pt,font=\scriptsize,align=left] {$...$};
            \draw (328.97,188.7) node [anchor=north west,inner sep=0.75pt,font=\footnotesize,align=left] {$\varphi_{R}(a_{2})$};
            \draw (294.93,95.43) node [anchor=north west,inner sep=0.75pt,align=left] {$\#$};
            \draw (294.47,159.23) node [anchor=north west,inner sep=0.75pt,align=left] {$\#$};
        \end{tikzpicture}
        \caption{The unravelling of $\mathsf{Arn}$ from~\cref{fig:arena-undecidability-ncrs} for a fixed $\sigma_0$.}
        \label{fig:arena-unravelling-undecidability-ncrs}
    \end{minipage}
\end{figure}

The two-player arena $\mathsf{Arn} = (V,E,\{0,1\},(V_0,V_1))$ is shown in \cref{fig:arena-undecidability-ncrs}. A play in this arena is constructed as follows: first, player~$0$ selects some arbitrary letter $a \in A$ starting from the initial vertex $v=\#$. Player~$1$ responds with one of the following three choices:
\begin{enumerate}
    \item move back to $\#$, forcing player~$0$ to choose another letter in $A$;
    \item choose $L_a$, which then forces a visit to a sequence of vertices corresponding to $\varphi_L(a)$ followed by a move into a $|B|$-clique, each vertex of which is an element of $B$ owned by player~$0$; or
    \item choose $R_a$, which visits vertices corresponding to $\varphi_R(a)$ before entering $B$.
\end{enumerate}
 
Any strategy of player~0 can be represented as a tree, depicted in Figure~\ref{fig:arena-unravelling-undecidability-ncrs}, and can be identified with a triple $(u, (u_i)_{i\geq 1}, (v_i)_{i\geq 1})$ where: $(i)$ $u=a_1a_2\dots\in A^\omega$ is the word formed by player~$0$ when player~$1$ plays $\#$ \emph{ad infinitum}, $(ii)$ for all $i\geq 1$, $u_i$ is the word formed by player~$0$ in the $B$-component of $\mathsf{Arn}$, after history $\#a_1\dots \#a_iL_{a_i}\varphi_L(a_i)$, and $(iii)$ for all $i\geq 1$, $v_i$ is the word formed by player~$0$ in the $B$-component, after history $\#a_1\dots \#a_iR_{a_i}\varphi_R(a_i)$.

The preference relation $\R_0$ is defined as the empty set $\varnothing$, so that whatever the strategy profile $(\sigma_0,\sigma_1)$ and play $\pi$, $\pi \R_0 {\outcomefrom{\bsigma}{\#}}$ does \emph{not} hold. Therefore, there exists a solution to the NCRS instance if and only if there exists a strategy $\sigma_0$ such that \emph{no} $\sigma_0$-fixed NE exists. The preference relation $\R_1$ for player~$1$ is defined in such a way that there exists a $\sigma_0$ without $\sigma_0$-fixed NEs if and only if there is a solution to our $\omega$-PCP instance. It essentially follows four rules (see the dotted arrows in Figure~\ref{fig:arena-unravelling-undecidability-ncrs}):
\begin{description}
    \item[$(i)$] the rightmost branch is preferred to the leftmost branch if $\varphi_L(a_1)u_1=\varphi_R(a_1)v_1$,
    \item[$(ii)$] the infinite branch visiting $L_{a_i}$ is preferred to the one visiting $L_{a_{i+1}}$ if $u_i = \varphi_L(a_{i+1})u_{i+1}$,
    \item[$(iii)$] the infinite branch visiting $R_{a_{i+1}}$ is preferred to the one visiting $R_{a_i}$ if $v_i = \varphi_R(a_{i+1})v_{i+1}$,
    \item[$(iv)$] the leftmost branch is always preferred to the middle branch.
\end{description}
In cases $(i)$-$(iii)$, if the condition does not hold, there is no preference. It is not difficult to show that $\R_1$ is $\omega$-automatic.

To conclude the sketch of proof, if there is a solution $a_1a_2\dots\in A^\omega$ to the $\omega$-PCP instance, then player~0 can enforce the non-existence of a $\sigma_0$-fixed NE by choosing $a_1a_2\dots$ as middle branch, and by choosing for all $i\geq 1$, $u_i = \varphi_L(a_{i+1}a_{i+2}\dots)$ and $v_i=\varphi_R(a_{i+1}a_{i+2}\dots)$. Hence, $u_i = \varphi_L(a_{i+1})u_{i+1}$, $v_i=\varphi_R(a_{i+1})v_{i+1}$, $\varphi_L(a_1)u_1 = \varphi_L(a_1a_2\dots) = \varphi_R(a_1a_2\dots)=\varphi_R(a_1)v_1$. So, player~$1$ always has an incentive to deviate in that case. 

Conversely, if there is a strategy $\sigma_0$, identified as a triple $(u=a_1a_2\dots, (u_i)_{i\geq 1}, (v_i)_{i\geq 1})$, with no $\sigma_0$-fixed Nash equilibria, then player~$1$ always has an incentive to deviate. Let us show that there is a solution to the $\omega$-PCP instance. By rule $(ii)$, we obtain that $u_1 = \varphi_L(a_2)u_2 = \varphi_L(a_2)\varphi_L(a_3)u_3$ and so on, ensuring that $u_1 = \varphi_L(a_2a_3\dots)$. Similarly, from rule $(iii)$ we obtain $v_1 = \varphi_R(a_2a_3\dots)$. Finally, from rule $(i)$, we get $\varphi_L(a_1a_2\dots) = \varphi_L(a_1)u_1 = \varphi_R(a_1)v_1 = \varphi_R(a_1a_2\dots)$. 
\end{proof}

We conclude with two verification problems: given a finite-memory strategy $\sigma_0$, verify whether it is a solution to the CRS (resp.\ NCRS) instance. For the NCRS problem, we recover decidability because it is simpler than deciding whether such a solution exists.

\begin{theorem}[restate=thmncrv,name=]
\label{thm:verification-rational-synthesis}
    Given a game $\game$, an initial vertex $v \in V$, a lasso $\pi \in \plays(v)$ and a strategy $\sigma_0 \in \Gamma_0$ defined by a Mealy machine, deciding whether $\sigma_0$ is a solution to the CRS (resp.\ NCRS) instance is \pspaceComplete{}.% As verif is decidable and synthesis is not, we cannot bound the memory size of solutions $\sigma_0$ for synthesis.
\end{theorem}

\begin{proof}
    Let $\aut{M}$ define $\sigma_0$. We form the polynomial-size product arena of $\game$ and $\aut{M}$, retaining only the moves prescribed by $\aut{M}$ at vertices of player~$0$. Let $p$ denote projection to the original vertex sequence. In the resulting game $\game'$, each follower's preference is lifted through $p$, while player~0 has the empty preference relation. Then the ordinary NEs of $\game'$ correspond to the $\sigma_0$-fixed NEs of $\game$. The original leader relation is retained only in the constraint $C = \{\rho' \mid \pi \R_0 p(\rho')\}$, recognized by a polynomial-size \DPW{} obtained from the relation automaton and $\pi$. Cooperative verification asks whether there is some NE outcome of $\game'$ in $C$; non-cooperative verification asks whether there is some NE outcome of $\game'$ in the complement of $C$. \Cref{thm:NEthreshold} and closure of PSPACE under complementation yield both upper bounds.

    The \pspaceHard{}ness is obtained by reduction from the NE existence problem. The reduction is the same as in \pspaceHard{}ness of \cref{thm:crs-pspace-ncrs-undec}: for the cooperative rational verification, we add a new player~$0$ who owns no vertices and has the relation $\mathord{\R_0} = V^\omega \times V^\omega$, while for the non-cooperative case, we set $\mathord{\R_0} = \varnothing$.
\end{proof}

\begin{remark}
    In \cref{thm:verification-rational-synthesis}, we limit our attention to strategies $\sigma_0$ defined by deterministic Mealy machines. In the context of rational verification, it is common to consider \emph{non-deterministic} Mealy machines, which serve as specifications defining a (possibly infinite) set of allowed strategies~\cite{BriceRB23-rational-verif,BruyereGrandmontRaskin24-concur,BruyereRaskinTamines22}; this problem is termed \emph{universal rational verification}. The concept of non-deterministic strategies was also studied under the name of permissive strategy, multi-strategy, or non-deterministic strategy, in the context of synthesis in two-player zero-sum games~\cite{BernetJW02,BouyerDMR09,PinchinatR05}.
    In our framework, the universal non-cooperative NE rational verification problem asks, given a lasso $\pi$ starting at $v$ and a non-deterministic Mealy machine $\aut{M}$, whether for all strategies $\sigma_0 \in \Gamma_0$ compatible with $\aut{M}$, for all $\sigma_0$-fixed NEs $\bsigma$ from $v$, it holds that $\pi \R_0 {\outcomefrom{\bsigma}{v}}$. This problem remains \pspaceComplete{}, with essentially the same proof as for the non-cooperative case of \cref{thm:verification-rational-synthesis} (see~\cite{BriceRB23-rational-verif} for the product construction).
    However, the universal cooperative NE rational verification becomes undecidable. This is because its negation asks whether there exists a strategy $\sigma_0$ of $\aut{M}$ such that for all $\sigma_0$-fixed NEs, it holds that $\pi \not\R_0 {\outcomefrom{\bsigma}{v}}$, which is precisely the NCRS problem (with a negated lasso constraint), proved to be undecidable in \cref{thm:crs-pspace-ncrs-undec}.
    \lipicsEnd
\end{remark}

\section{Conclusion}
\label{section:conclusion-future-work}

We have studied multiplayer games with $\omega$-automatic preference relations. In the two-player zero-sum setting, we have obtained complexity bounds for various \emph{threshold}- and \emph{optimality}-related problems (\cref{thm:threshold-problem,thm:optimal-strategy}). Then, we have investigated decision problems for NEs in the non-zero-sum multiplayer setting. Since NEs do not necessarily exist in this context, we first addressed the decision problem of their existence: we have established a tight \pspace{} complexity bound for deciding whether there exists an NE, even if it is required to satisfy additional $\omega$-regular constraints (specified as threshold lassos, LTL formulas, or \APWs{}). Finally, we have analyzed the problem of rational synthesis. In the cooperative case, we have derived a tight \pspace{} complexity bound, whereas in the non-cooperative case, we have proved that the synthesis problem is undecidable. In contrast, the corresponding rational verification problems are \pspaceComplete{} in both cases.

In sequential games, NEs are known to suffer from the issue of non-credible threats~\cite{OsborneRubinstein1994,selten1965}. Perhaps the most natural future direction is to study the decision problems of this paper with more robust rationality notions, such as \emph{subgame perfect equilibria}~\cite{selten1965} or \emph{admissible strategies}~\cite{Berwanger07,BrenguierRS14} that are immune to non-credible threats. Let us mention that the notion of \emph{secure equilibria}, a class of NEs with lexicographically ordered objectives~\cite{BrihayeBPG13,ChatterjeeHJ06}, is already covered by this paper, as the latter can be expressed as $\omega$-automatic preferences. 

Although we have demonstrated that $\omega$-automatic relations---formalized as $\omega$-regular languages over the alphabet $V \times V$---are highly expressive and broadly applicable, they admit a natural extension to \emph{rational relations}, that is, subsets of $V^\omega \times V^\omega$ recognizable by finite-state transducers. An open problem is whether the decision problems studied in this work remain decidable in this more expressive framework; it is plausible that a substantial subset of them becomes undecidable. Another natural extension is the usage of \emph{more powerful} automata models to define preference relations, such as (visibly) pushdown automata. For example, we conjecture that we can capture total-payoff or energy objectives with pushdown automata, which are not expressible with $\omega$-automatic preference relations. Conversely, it is equally crucial to identify and analyze natural subclasses of $\omega$-automatic preference relations for which non-cooperative rational synthesis regains decidability. A natural candidate for this would be $\omega$-recognizable relations~\cite{LodingSpinrath2019}.

\bibliography{CSL_long_version}

\appendix

\section{Proof of \texorpdfstring{\cref{thm:threshold-problem}}{Theorem~\ref{thm:threshold-problem}}}
\label{app:thresholdproblems}

\thmthresholdproblem*

\begin{proof}
    The proof consists of three parts, one for each problem studied.

    \subparagraph*{First problem}
    Let $\pi$ be a lasso in $\plays(v)$. Solving the threshold problem for the instance $\pi$ is equivalent to checking whether $\pi \in \val_1(v)$, that is, $\pi$ is accepted by the \APW{} $\aut{B}$ in \cref{thm:values-automatic}. It is classical that the \APW{} membership problem can be reduced to solving a (two-player zero-sum) parity game~\cite{handbook-Automata-WilkeSchewe21}, and that parity games have memoryless winning strategies~\cite{Zielonka1998-parity-game}. Here, with the notation of the proof of \cref{thm:values-automatic}, the set of vertices $S'$ of this parity game $\aut{H}$ is the product of $S$ and $\pi$, such that the symbols of $\pi$ constrain the first component of the states of $S$. This leads to a set $S'$ of size $|V| \cdot |\aut{A}| \cdot |\pi|$. Moreover, a memoryless winning strategy in $\aut{H}$ becomes a finite-memory strategy in $\aut{G}$ such that its memory size is $|\aut{A}| \cdot |\pi|$.

    To prove the parity-hardness of the threshold problem, we use a reduction from the problem of deciding whether player~$1$ has a winning strategy in a parity game. This reduction is very close to the reduction proposed in~\cite[Theorem~4]{BruyereGrandmontRaskin25-mfcs}. Let $\mathcal{H}$ be a parity game with players~$1$ and $2$, an arena $\arena$ with $V$ as set of vertices, an initial vertex $v_0$, and a priority function $\alpha$. We construct a new game $\game = (\arena', \R)$ with the same players, whose arena $\arena'$ is a copy of $\arena$ with an additional vertex $v_0'$ owned by player~$1$, with $v_0$ and itself as successors (see \cref{fig:outcome-checking-parity-reduction}). Given $V' = V \cup \{v'_0\}$, the preference relation $\R$ is defined as follows: $x \R y$ if $x = v_0'^\omega$ and $y = (v_0')^m y'$, where $m \geq 0$ and $y'$ is a play in $\aut{H}$ starting at $v_0$ and satisfying the parity condition. A \DPW{} $\aut{A}$ accepting $\R$ is depicted in \cref{fig:outcome-checking-parity-reduction}, it is constructed with a copy of the arena $\arena$ and a new state $q_0$ with priority~1. Finally, we choose the lasso $(v_0')^\omega$ as threshold. 

    Let us prove that the (polynomial) reduction is correct. Suppose first that there exists a strategy $\sigma_1$ that is a solution to the threshold problem instance. This means that for each play $\rho \in \plays(v'_0)$ consistent with $\sigma_1$, we have $\pi \R \rho$. Thus, by definition of $\R$, $\rho$ is equal to $(v'_0)^m\rho'$ with $\rho'$ a winning play in $\mathcal{H}$. Hence, from $\sigma_1$, we can derive a winning strategy for player~1 in the parity game $\aut{H}$. The other direction is proved similarly: if player~$1$ has a winning strategy in $\aut{H}$, then transferring this strategy to $\game$ gives a solution to the threshold problem instance.

    \begin{figure}[t] 
        \centering
        \begin{tikzpicture}[x=0.75pt,y=0.75pt,yscale=-1,xscale=1] 
            \draw (258.31,117.01) .. controls (258.31,112.38) and (261.86,108.63) .. (266.25,108.63) .. controls (270.64,108.63) and (274.19,112.38) .. (274.19,117.01) .. controls (274.19,121.63) and (270.64,125.38) .. (266.25,125.38) .. controls (261.86,125.38) and (258.31,121.63) .. (258.31,117.01) -- cycle;
            \draw (266.51,96.18) -- (266.31,105.63);
            \draw [shift={(266.25,108.63)}, rotate = 271.21,fill={rgb,255:red,0;green,0;blue,0},line width=0.08,draw opacity=0] (5.36,-2.57) -- (0,0) -- (5.36,2.57) -- (3.56,0) -- cycle;
            \draw [fill={rgb, 255:red, 0; green, 0; blue, 0 }  ,fill opacity=0.03 ,dash pattern={on 4.5pt off 4.5pt}] (147.06,112.15) .. controls (147.06,107.43) and (150.88,103.61) .. (155.59,103.61) -- (189.65,103.61) .. controls (194.37,103.61) and (198.19,107.43) .. (198.19,112.15) -- (198.19,141.83) .. controls (198.19,146.55) and (194.37,150.37) .. (189.65,150.37) -- (155.59,150.37) .. controls (150.88,150.37) and (147.06,146.55) .. (147.06,141.83) -- cycle;
            \draw (116.81,122.81) .. controls (116.81,118.42) and (120.36,114.87) .. (124.74,114.87) .. controls (129.12,114.87) and (132.67,118.42) .. (132.67,122.81) .. controls (132.67,127.19) and (129.12,130.74) .. (124.74,130.74) .. controls (120.36,130.74) and (116.81,127.19) .. (116.81,122.81) -- cycle;
            \draw (132.67,122.81) -- (155.6,122.89);
            \draw [shift={(158.6,122.9)}, rotate = 180,fill={rgb,255:red,0;green,0;blue,0},line width=0.08,draw opacity=0] (5.36,-2.57) -- (0,0) -- (5.36,2.57) -- (3.56,0) -- cycle;
            \draw (158.81,122.81) .. controls (158.81,118.42) and (162.36,114.87) .. (166.74,114.87) .. controls (171.12,114.87) and (174.67,118.42) .. (174.67,122.81) .. controls (174.67,127.19) and (171.12,130.74) .. (166.74,130.74) .. controls (162.36,130.74) and (158.81,127.19) .. (158.81,122.81) -- cycle;
            \draw (262.07,124.24) .. controls (257.55,142.42) and (273.21,142.97) .. (270.16,126.86);
            \draw [shift={(269.51,124.18)}, rotate = 73.97,fill={rgb,255:red,0;green,0;blue,0},line width=0.08,draw opacity=0] (5.36,-2.57) -- (0,0) -- (5.36,2.57) -- (3.56,0) -- cycle;
            \draw (274.19,117.01) -- (321.31,116.7);
            \draw [shift={(324.31,116.68)}, rotate = 180,fill={rgb,255:red,0;green,0;blue,0},line width=0.08,draw opacity=0] (5.36,-2.57) -- (0,0) -- (5.36,2.57) -- (3.56,0) -- cycle;
            \draw (121.37,115.14) .. controls (119.1,98.13) and (130.4,96.42) .. (128.28,112.04);
            \draw [shift={(127.77,114.94)}, rotate = 282.09,fill={rgb,255:red,0;green,0;blue,0},line width=0.08,draw opacity=0] (5.36,-2.57) -- (0,0) -- (5.36,2.57) -- (3.56,0) -- cycle;
            \draw [fill={rgb, 255:red, 0; green, 0; blue, 0 }  ,fill opacity=0.03 ,dash pattern={on 4.5pt off 4.5pt}] (316.76,112.94) .. controls (316.76,107.62) and (321.07,103.31) .. (326.39,103.31) -- (388.23,103.31) .. controls (393.55,103.31) and (397.86,107.62) .. (397.86,112.94) -- (397.86,146.45) .. controls (397.86,151.77) and (393.55,156.08) .. (388.23,156.08) -- (326.39,156.08) .. controls (321.07,156.08) and (316.76,151.77) .. (316.76,146.45) -- cycle;
            \draw (323.97,117.31) .. controls (323.97,112.92) and (327.53,109.37) .. (331.91,109.37) .. controls (336.29,109.37) and (339.84,112.92) .. (339.84,117.31) .. controls (339.84,121.69) and (336.29,125.24) .. (331.91,125.24) .. controls (327.53,125.24) and (323.97,121.69) .. (323.97,117.31) -- cycle;
            \draw (378.31,120.16) .. controls (378.31,115.78) and (381.86,112.23) .. (386.24,112.23) .. controls (390.62,112.23) and (394.17,115.78) .. (394.17,120.16) .. controls (394.17,124.54) and (390.62,128.09) .. (386.24,128.09) .. controls (381.86,128.09) and (378.31,124.54) .. (378.31,120.16) -- cycle;
            \draw (342.55,126.01) .. controls (347.78,119.66) and (361.55,120.07) .. (375.42,120.15);
            \draw [shift={(378.31,120.16)}, rotate = 180,fill={rgb,255:red,0;green,0;blue,0},line width=0.08,draw opacity=0] (5.36,-2.57) -- (0,0) -- (5.36,2.57) -- (3.56,0) -- cycle;
            \draw (117.2,114.53) node [anchor=north west,inner sep=0.75pt,font=\small,align=left] {$v'_{0}$};
            \draw (235.79,101.77) node [anchor=north west,inner sep=0.75pt,align=left] {$\aut{A}$};
            \draw (180.16,131.67) node [anchor=north west,inner sep=0.75pt,align=left] {$\aut{H}$};
            \draw (259.5,112.63) node [anchor=north west,inner sep=0.75pt,font=\small,align=left] {$q_{0}$};
            \draw (159.6,119.13) node [anchor=north west,inner sep=0.75pt,font=\small,align=left] {$v_{0}$};
            \draw (96.09,101.27) node [anchor=north west,inner sep=0.75pt,align=left] {$\game$};
            \draw (275.23,100.8) node [anchor=north west,inner sep=0.75pt,font=\small]  {$(v'_{0},v_{0})$};
            \draw (245.23,137) node [anchor=north west,inner sep=0.75pt,font=\small,align=left] {$(v'_{0},v'_{0})$};
            \draw (378.46,138.17) node [anchor=north west,inner sep=0.75pt,align=left] {$\aut{H}$};
            \draw (324.76,113.63) node [anchor=north west,inner sep=0.75pt,font=\small,align=left] {$v_{0}$};
            \draw (341.87,124.2) node [anchor=north west,inner sep=0.75pt,font=\small,align=left] {$(v'_{0},v)$};
            \draw (381.7,117.28) node [anchor=north west,inner sep=0.75pt,font=\small,align=left] {$v$};
        \end{tikzpicture}
        \caption{The game $\game$ and the \DPW{} $\aut{A}$ accepting $\R$ for the \parityHard{}ness of \cref{thm:threshold-problem}.}
        \label{fig:outcome-checking-parity-reduction}
    \end{figure}

    \subparagraph*{Second problem}
    We have to decide whether $\val_1(v) \neq \varnothing$. The \pspace-membership follows from the recognizability of $\val_1(v)$ by some polynomial size \APW{} (\cref{thm:values-automatic}) and the \pspace{}-membership of \APW{} non-emptiness problem~\cite{DaxKlaedtke08,Leucker-Sanchez-2010}: one can construct on the fly an \NBW{} of exponential size~\cite{DaxKlaedtke08}, while checking its non-emptiness with the classical \nl{} algorithm. If $\val_1(v) \neq \varnothing$, as $\val_1(v)$ is $\omega$-regular, there exists a lasso $\pi$ in $\val_1(v)$.

    To prove \pspaceHard{}ness, we reduce the \NBW{} universality problem, which is \pspaceComplete{}~\cite{handbook-automata-kupferman18}, to the negation of our problem, i.e., whether $\val_1(v) = \varnothing$. Let $\aut{B} = (Q,q_0,\delta,F)$ be an \NBW{} over the alphabet $\Sigma$. We define a two-player zero-sum game $\game = (\arena, \R)$. Its arena $\arena$, depicted in \cref{fig:reduction-verifthreshold}, is constructed as follows. The set $V$ of vertices is equal to $\{v_0,v_1\} \cup \Sigma \cup Q$, such that $V_1 = \{v_0\}$ and $V_2 = V \setminus V_1$. From $v_0$, player~$1$ can move to $v_1$ or $q_0 \in Q$; from $v_1$ or a letter $a \in \Sigma$ (resp.\ a state $q \in Q$), player~$2$ can move to any letter $a' \in \Sigma$ (resp.\ any state $q' \in Q$).
    The relation $\R$ is defined as follows: $x \R y$ holds if $x = v_0 v_1 w$ with $w \in \Sigma^\omega$, $y = v_0 r$ with $r \in Q^\omega$, and $r$ is not an accepting run on $w$ (i.e., $r$ is either not a run on $w$ or it is a run but not accepting). This relation is accepted by the \DPW{} $\aut{A}$ depicted in \cref{fig:reduction-verifthreshold}, such that there exists a transition $q \xrightarrow{(a,q')} q'$ in $\aut{A}$ if and only if there exists a transition $q \xrightarrow{a} q'$ in $\aut{B}$, and a state $q$ has priority 1 (resp.\ 0) in $\aut{A}$ if and only if $q \in F$ (resp.\ $q \not\in F$) in $\aut{B}$. Moreover, there is a state $\top$ of priority $0$ with a loop labeled $(a,q)$ for all $a \in \Sigma$ and $q \in Q$, that is reached by transitions as shown in the figure, so as to accept pairs $(v_0v_1w,v_0r)$ such that $r \in Q^\omega$ is not a run on $w$. Missing transitions lead to a sink state with priority $1$. By definition of $\R$, observe that
    \[
    \val_1(v_0) = \{v_0v_1w \in v_0v_1\Sigma^\omega \mid \forall r \in Q^\omega, r \text{ is not an accepting run for } w\} = v_0 v_1 (\Sigma^\omega \setminus \lang(\aut{B})).
    \]
    This holds because player~$1$ can move to $q_0$ to let player~$2$ generate an arbitrary word $r \in q_0Q^\omega$. Thus, $\val_1(v_0) = \varnothing$ if and only if $\lang(\aut{B}) = \Sigma^\omega$. This equivalence establishes the correctness of the polynomial reduction.

\begin{figure}
    \centering
    \begin{tikzpicture}[x=0.75pt,y=0.75pt,yscale=-1,xscale=1]
        \draw  [fill={rgb,255:red,155;green,155;blue,155},fill opacity=0.1 ,dash pattern={on 4.5pt off 4.5pt},line width=0.75]  (212.12,96.75) .. controls (212.12,92.29) and (215.74,88.67) .. (220.21,88.67) -- (251.25,88.67) .. controls (255.71,88.67) and (259.33,92.29) .. (259.33,96.75) -- (259.33,121.02) .. controls (259.33,125.49) and (255.71,129.11) .. (251.25,129.11) -- (220.21,129.11) .. controls (215.74,129.11) and (212.12,125.49) .. (212.12,121.02) -- cycle;
        \draw (166.23,134.71) -- (181.49,126.49);
        \draw [shift={(184.13,125.07)}, rotate = 151.7,fill={rgb, 255:red, 0; green, 0; blue, 0 }  ,line width=0.08,draw opacity=0] (5.36,-2.57) -- (0,0) -- (5.36,2.57) -- (3.56,0) -- cycle;
        \draw (167.56,140.64) -- (211.26,153.02);
        \draw [shift={(214.15,153.84)}, rotate = 195.82,fill={rgb, 255:red, 0; green, 0; blue, 0 }  ,line width=0.08,draw opacity=0] (5.36,-2.57) -- (0,0) -- (5.36,2.57) -- (3.56,0) -- cycle;
        \draw (139.15,138.33) -- (149.15,138.21);
        \draw [shift={(152.15,138.17)}, rotate = 179.29,fill={rgb, 255:red, 0; green, 0; blue, 0 }  ,line width=0.08,draw opacity=0] (5.36,-2.57) -- (0,0) -- (5.36,2.57) -- (3.56,0) -- cycle;
        \draw  [fill={rgb,255:red,155;green,155;blue,155},fill opacity=0.1 ,dash pattern={on 4.5pt off 4.5pt},line width=0.75]  (212.45,146.83) .. controls (212.45,142.43) and (216.02,138.87) .. (220.41,138.87) -- (251.71,138.87) .. controls (256.1,138.87) and (259.67,142.43) .. (259.67,146.83) -- (259.67,170.71) .. controls (259.67,175.1) and (256.1,178.67) .. (251.71,178.67) -- (220.41,178.67) .. controls (216.02,178.67) and (212.45,175.1) .. (212.45,170.71) -- cycle;
        \draw  [fill={rgb,255:red,155;green,155;blue,155},fill opacity=0.1 ,dash pattern={on 4.5pt off 4.5pt},line width=0.75]  (408.15,100.38) .. controls (408.15,92.36) and (414.65,85.86) .. (422.67,85.86) -- (510.1,85.86) .. controls (518.12,85.86) and (524.62,92.36) .. (524.62,100.38) -- (524.62,143.95) .. controls (524.62,151.97) and (518.12,158.47) .. (510.1,158.47) -- (422.67,158.47) .. controls (414.65,158.47) and (408.15,151.97) .. (408.15,143.95) -- cycle;
        \draw (290.15,106.84) .. controls (290.15,102.57) and (293.61,99.12) .. (297.87,99.12) .. controls (302.14,99.12) and (305.6,102.57) .. (305.6,106.84) .. controls (305.6,111.11) and (302.14,114.56) .. (297.87,114.56) .. controls (293.61,114.56) and (290.15,111.11) .. (290.15,106.84) -- cycle;
        \draw (277.15,107) -- (287.15,106.88);
        \draw [shift={(290.15,106.84)}, rotate = 179.29,fill={rgb, 255:red, 0; green, 0; blue, 0 }  ,line width=0.08,draw opacity=0] (5.36,-2.57) -- (0,0) -- (5.36,2.57) -- (3.56,0) -- cycle;
        \draw (305.6,106.84) -- (345.6,106.69);
        \draw [shift={(348.6,106.68)}, rotate = 179.79,fill={rgb, 255:red, 0; green, 0; blue, 0 }  ,line width=0.08,draw opacity=0] (5.36,-2.57) -- (0,0) -- (5.36,2.57) -- (3.56,0) -- cycle;
        \draw (348.6,106.68) .. controls (348.6,102.41) and (352.05,98.96) .. (356.32,98.96) .. controls (360.58,98.96) and (364.04,102.41) .. (364.04,106.68) .. controls (364.04,110.95) and (360.58,114.4) .. (356.32,114.4) .. controls (352.05,114.4) and (348.6,110.95) .. (348.6,106.68) -- cycle;
        \draw (364.26,107.17) -- (407.26,107.02);
        \draw [shift={(410.26,107.01)}, rotate = 179.8,fill={rgb, 255:red, 0; green, 0; blue, 0 }  ,line width=0.08,draw opacity=0] (5.36,-2.57) -- (0,0) -- (5.36,2.57) -- (3.56,0) -- cycle;
        \draw (197.24,113.81) -- (218.77,100.69);
        \draw [shift={(221.33,99.13)}, rotate = 148.66,fill={rgb, 255:red, 0; green, 0; blue, 0 }  ,line width=0.08,draw opacity=0] (5.36,-2.57) -- (0,0) -- (5.36,2.57) -- (3.56,0) -- cycle;
        \draw (410.26,107.01) .. controls (410.26,102.75) and (413.72,99.29) .. (417.99,99.29) .. controls (422.25,99.29) and (425.71,102.75) .. (425.71,107.01) .. controls (425.71,111.28) and (422.25,114.74) .. (417.99,114.74) .. controls (413.72,114.74) and (410.26,111.28) .. (410.26,107.01) -- cycle;
        \draw (417.6,130.35) .. controls (417.6,126.08) and (421.05,122.62) .. (425.32,122.62) .. controls (429.58,122.62) and (433.04,126.08) .. (433.04,130.35) .. controls (433.04,134.61) and (429.58,138.07) .. (425.32,138.07) .. controls (421.05,138.07) and (417.6,134.61) .. (417.6,130.35) -- cycle;
        \draw (503.26,130.35) .. controls (503.26,126.08) and (506.72,122.62) .. (510.99,122.62) .. controls (515.25,122.62) and (518.71,126.08) .. (518.71,130.35) .. controls (518.71,134.61) and (515.25,138.07) .. (510.99,138.07) .. controls (506.72,138.07) and (503.26,134.61) .. (503.26,130.35) -- cycle;
        \draw (433.04,130.35) -- (500.26,130.35);
        \draw [shift={(503.26,130.35)}, rotate = 180,fill={rgb, 255:red, 0; green, 0; blue, 0 }  ,line width=0.08,draw opacity=0] (5.36,-2.57) -- (0,0) -- (5.36,2.57) -- (3.56,0) -- cycle;
        \draw (152.15,138.17) .. controls (152.15,133.91) and (155.61,130.45) .. (159.87,130.45) .. controls (164.14,130.45) and (167.6,133.91) .. (167.6,138.17) .. controls (167.6,142.44) and (164.14,145.9) .. (159.87,145.9) .. controls (155.61,145.9) and (152.15,142.44) .. (152.15,138.17) -- cycle;
        \draw (197.47,120.4) -- (220,120.75);
        \draw [shift={(223,120.8)}, rotate = 180.9,fill={rgb, 255:red, 0; green, 0; blue, 0 }  ,line width=0.08,draw opacity=0] (5.36,-2.57) -- (0,0) -- (5.36,2.57) -- (3.56,0) -- cycle;
        \draw (183.84,113.81) -- (197.24,113.81) -- (197.24,127.21) -- (183.84,127.21) -- cycle;
        \draw (214.84,149.47) -- (228.24,149.47) -- (228.24,162.87) -- (214.84,162.87) -- cycle;
        \draw (367.57,170.6) .. controls (408.25,170.88) and (424.25,169) .. (425.32,138.07);
        \draw [shift={(364.4,170.57)}, rotate = 0.51,fill={rgb, 255:red, 0; green, 0; blue, 0 }  ,line width=0.08,draw opacity=0] (5.36,-2.57) -- (0,0) -- (5.36,2.57) -- (3.56,0) -- cycle;
        \draw (348.95,170.57) .. controls (348.95,166.31) and (352.41,162.85) .. (356.67,162.85) .. controls (360.94,162.85) and (364.4,166.31) .. (364.4,170.57) .. controls (364.4,174.84) and (360.94,178.3) .. (356.67,178.3) .. controls (352.41,178.3) and (348.95,174.84) .. (348.95,170.57) -- cycle;
        \draw (350.72,176.18) .. controls (328.54,180.57) and (326.18,164.89) .. (347.26,167.17);
        \draw [shift={(350.02,167.55)}, rotate = 189.46,fill={rgb, 255:red, 0; green, 0; blue, 0 }  ,line width=0.08,draw opacity=0] (5.36,-2.57) -- (0,0) -- (5.36,2.57) -- (3.56,0) -- cycle;
        \draw (356.32,114.4) -- (356.65,159.85);
        \draw [shift={(356.67,162.85)}, rotate = 269.58,fill={rgb, 255:red, 0; green, 0; blue, 0 }  ,line width=0.08,draw opacity=0] (5.36,-2.57) -- (0,0) -- (5.36,2.57) -- (3.56,0) -- cycle;

        \draw (228.12,103.8) node [anchor=north west,inner sep=0.75pt,font=\normalsize,align=left] {$\Sigma $};
        \draw (229.79,157.47) node [anchor=north west,inner sep=0.75pt,font=\normalsize,align=left] {$Q$};
        \draw (504.12,94.8) node [anchor=north west,inner sep=0.75pt,font=\normalsize,align=left] {$\mathcal{B}$};
        \draw (306.3,91.4) node [anchor=north west,inner sep=0.75pt,font=\footnotesize,align=left] {$(v_{0} ,v_{0})$};
        \draw (363.97,92.07) node [anchor=north west,inner sep=0.75pt,font=\footnotesize,align=left] {$(v_{1} ,q_{0})$};
        \draw (411.72,102.33) node [anchor=north west,inner sep=0.75pt,font=\footnotesize,align=left] {$q_{0}$};
        \draw (421.06,126.67) node [anchor=north west,inner sep=0.75pt,font=\footnotesize,align=left] {$q$};
        \draw (505.72,124.67) node [anchor=north west,inner sep=0.75pt,font=\footnotesize,align=left] {$q'$};
        \draw (433.5,115.9) node [anchor=north west,inner sep=0.75pt,font=\footnotesize,align=left] {\begin{minipage}[lt]{49.91pt}\setlength\topsep{0pt}
        \begin{center}
        $(a,q')$
        \end{center}
        if $q'\in \delta (q,a)$
        \end{minipage}};
        \draw (153.39,134.67) node [anchor=north west,inner sep=0.75pt,font=\footnotesize,align=left] {$v_{0}$};
        \draw (184.06,117) node [anchor=north west,inner sep=0.75pt,font=\footnotesize,align=left] {$v_{1}$};
        \draw (215.06,152.67) node [anchor=north west,inner sep=0.75pt,font=\footnotesize,align=left] {$q_{0}$};
        \draw (350.7,166.8) node [anchor=north west,inner sep=0.75pt,font=\footnotesize,align=left] {$\top$};
        \draw (366.4,171.57) node [anchor=north west,inner sep=0.75pt,font=\footnotesize,align=left] {$(a,q')$ if $q'\notin \delta ( q,a)$};
        \draw (295.8,165.97) node [anchor=north west,inner sep=0.75pt,font=\footnotesize,align=left] {$\Sigma \times Q$};
        \draw (358.97,123.7) node [anchor=north west,inner sep=0.75pt,font=\footnotesize,align=left] {$(v_{1} ,q)$\\$q\neq q_{0}$};
    \end{tikzpicture}
    \caption{The game $\game$ and the \DPW{} $\aut{A}$ accepting $\R$ for the \pspace-hardness of \cref{thm:threshold-problem}.}
    \label{fig:reduction-verifthreshold}
\end{figure}

    \subparagraph*{Third problem}
    Let $\aut{M}$ be a Mealy machine defining $\sigma_1$. To prove \nl{}-membership, as the class \nl{} is closed under complementation, we provide an algorithm to decide whether $\sigma_1$ is not a solution to the threshold problem instance for $\pi$. That is, to check whether the language $L = \{\rho \in \plays_{\sigma_1}(v) \mid \pi \notR \rho\}$ is non-empty. From $\arena$, $\aut{M}$, $\pi$, and $\aut{A}$ (with the priorities increased by 1 to deal with~$\not\R$), it is easy to construct on the fly a \DPW{} $\aut{B}$ accepting $L$, whose size is polynomial, while performing a non-emptiness test in \nl{}~\cite{handbook-automata-kupferman18}.

    To establish \nlHard{}ness, we use a reduction from the \DBW{} universality problem, an \nlComplete{} problem~\cite{handbook-automata-kupferman18}. From a \DBW{} $\aut{B}$ over the alphabet $\Sigma$ and seen as a \DPW{}, we construct a game $\game = (\arena,\R)$ as follows. The set $V$ of vertices of $\arena$ is equal to $\{v_0,v_1\} \cup \Sigma$ such that $V_1 = \{v_0\}$ and $V_2 = V \setminus \{v_0\}$. 
    The set of edges is equal to $\{(v_0,v_0),(v_0,v_1)\} \cup \{(v_1,a) \mid a \in \Sigma\} \cup \Sigma \times \Sigma$. Hence, in $v_0$, player~$1$ can either loop on $v_0$ or go to $v_1$. From $v_1$, player~2 can generate any word of $\Sigma^\omega$. We define the relation~$\R$ such that $x \R y$ if $x = v_0^\omega$ and $y \in v_0v_1 \lang(\aut{B})$. Observe that $\R$ is $\omega$-automatic and is accepted by a \DPW{} with a size linear in $|\aut{B}|$. Finally, we set $\pi = v_0^\omega$ and $\sigma_1(v_0) = v_1$. Therefore, we have $\lang(\aut{B}) = \Sigma^\omega$ if and only if $\sigma_1$ is a solution to the threshold problem instance for $\pi$. This completes the proof since the proposed reduction is logspace.
\end{proof}

\section{Proof of \texorpdfstring{\cref{thm:optimal-strategy}}{Theorem~\ref{thm:optimal-strategy}}}
\label{app:optimalstrategies}

\thmoptimalstrategy*

\begin{proof}
    The proof is in three parts.

    \subparagraph*{First problem}
    We begin with the first result and establish the \pspace{}-membership. Let $\aut{M}$ be a Mealy machine defining $\sigma_1$ and $v$ be a vertex. We consider the language
    \[
    L = \{(\pi,\rho) \in V^\omega \times V^\omega \mid (\pi \in \val_1(v) \land \rho \in \plays_{\sigma_1}(v)) \Rightarrow \pi \R \rho\}.
    \]
    As alternating automata are closed under Boolean operations, $L$ is accepted by an \APW{} constructed via the product of $\game$, $\aut{M}$, $\aut{A}$, and the \APW{} $\aut{B}$ accepting $\val_1(v)$, with a size linear in those components~\cite{handbook-automata-kupferman18}. Moreover, the language $L$ is universal, i.e., $L = V^\omega \times V^\omega$, if and only if $\sigma_1$ is optimal from $v$. Since the \APW{} universality problem is in \pspace{} (as done for the non-emptiness for \APWs{} in the proof of \cref{thm:threshold-problem}), the upper bound follows.

    For the \pspaceHard{}ness, we use the same reduction as done in the proof of the second result of \cref{thm:threshold-problem}. Given an \NBW{} $\aut{B}$, we construct the same game $\game = (\arena, \R)$ as depicted in \cref{fig:reduction-verifthreshold} and we consider the memoryless strategy $\sigma_1 \in \Gamma_1$ such that $\sigma_1(v_0) = v_1$. Plays $\rho$ consistent with $\sigma_1$ have the form $v_0 v_1 w$ with $w \in \Sigma^\omega$, such that there exists no $\pi$ with  $\pi \R \rho$. Consequently, $\sigma_1$ is (vacuously) optimal if and only if $\val_1(v_0) = \varnothing$, which holds if and only if $\lang(\aut{B}) = \Sigma^\omega$.

    \subparagraph*{Second problem, membership}
    We now proceed to the proof of the second result and first show the \twoExptime-membership. By definition, an optimal strategy $\sigma_1^* \in \Gamma_1$ from $v$ must ensure that any play $\rho \in \plays(v)$ consistent with $\sigma_1^*$ belongs to the language:
    \[
        L = \{\rho \in \plays(v) \mid \forall \pi \in \val_1(v), \pi \R \rho\}.
    \]
    Therefore, deciding the existence of an optimal strategy reduces to deciding whether player~$1$ has a winning strategy in the game played on $\arena$ with the objective $L$. We construct a \DPW{} $\aut{C}$ accepting $L$ through the following steps:
    \begin{enumerate}
        \item We consider the language $L_1 = \{(\pi,\rho) \in V^\omega \times V^\omega \mid \pi \not\in \val_1(v) ~\vee~ \pi \R \rho\}$. As \APWs{} are closed under Boolean operations, we construct an \APW{} $\aut{C}_1$ accepting the language $L_1$ from the complement of the \APW{} $\aut{B}$ of \cref{thm:values-automatic} and the \DPW{} $\aut{A}$ accepting $\R$. The automaton $\aut{C}_1$ has size polynomial in $|V|$ and $|\aut{A}|$, and index $|\alpha|+1$.

        \item We transform the \APW{} $\aut{C}_1$ into an \NBW{} $\aut{C}_2$ accepting the complement of $L_1$ equal to $L_2 = \{(\pi,\rho) \in V^\omega \times V^\omega \mid \pi \in \val_1(v) \land \pi \notR \rho\}$. The size of $\aut{C}_2$ is exponential in $|\aut{C}_1|$ (specifically $2^{\bigO{nk\log(n)}}$, with $n$ the size and $k$ the index of $\aut{C}_1$~\cite{DaxKlaedtke08}). Therefore, $\aut{C}_2$ has a size exponential in $|V|$, $|\aut{A}|$, and $|\alpha|$.

        \item We project $\aut{C}_2$ onto the $\rho$-component, to obtain an \NBW{} $\aut{C}_3$ accepting the language $L_3 = \{\rho \in V^\omega \mid \exists \pi \in \val_1(v), \pi \notR \rho\}$, with the same size as $\aut{C}_2$.

        \item To get the required \DPW{} $\aut{C}$ for $L$, we transform $\aut{C}_3$ into a \DPW{} using the Safra-Piterman construction~\cite{Piterman07},
        and take its complement. The resulting automaton $\aut{C}$ has size $2^{\bigO{n \log(n)}}$ and index $O(n)$ where $n$ is the size of $\aut{C}_3$. Therefore, $\aut{C}$ has a size (resp.\ an index) doubly exponential (resp.\ exponential) in $|V|$, $|\aut{A}|$, and $|\alpha|$.
    \end{enumerate}

    To decide whether player~$1$ has a winning strategy in the game played on $\arena$ with the objective~$L$, we finally construct the product game $\game' = \arena \times \aut{C}$, which is a parity game whose priority function is given by the one of $\aut{C}$. Recall that parity games can be solved in quasi-polynomial time (specifically in $n^{\bigO{\log k}}$ where $n$ is the size of the game and $k$ its index~\cite{Calude-Jain-Khoussainov-Li-Stephan2022}). Therefore, due to the size and index of $\aut{C}$, the game $\game'$ can be solved in  \twoExptime{}, leading to a \twoExptime-membership for deciding the existence of an optimal strategy.

    Notice that as parity games can be solved with memoryless strategies, if an optimal strategy $\sigma_1^*$ exists, then there is one that is memoryless in $\game'$, thus with a memory size doubly exponential in $|V|$, $|\aut{A}|$, and $|\alpha|$.

    \subparagraph*{Second problem, hardness}
    It remains to prove \pspaceHard{}ness for the second problem. We use a reduction from the non-emptiness problem of the intersection of $n$ \DFAs{} $\aut{D}_i = (Q_i, q_0^i, \delta_i,F_i)$, over a common alphabet $\Sigma$, which is known to be \pspaceComplete{}~\cite{Kozen77}. We assume that $\lang(\aut{D}_i) \neq \varnothing$ for each $i$ as otherwise the intersection of the $n$ \DFAs{} is empty. This can be checked in polynomial time.

    Given the \DFAs{}, we construct a one-player game $\game = (\arena,\R)$ with $V_2=\varnothing$ as follows. For simplicity, the arena $\arena$ is edge-labelled over some alphabet $\Sigma'$ and the relation $\R$ is expressed as a language over $(\Sigma' \times \Sigma')^\omega$. The alphabet is $\Sigma' = \{0,1,\dots,n,\bot\}\cup\Sigma$. Then, the set of vertices is $V = V_1 = \{v_0,v\}$. For all $i\in\{0,1,\dots,n\}$, there is an edge from $v_0$ to $v$ labeled $i$, and for all $a\in\Sigma\cup\{\bot\}$, there is an edge from $v$ to $v$ labeled $a$. So, a strategy for player~$1$ from $v_0$ is an infinite word of the form $i z$ with $i\in \{0,1,\dots,n\}$ and $z\in(\Sigma\cup \{\bot\})^\omega$. The preference relation is defined
    as $x \R y$ whenever $x = i\bot^\omega$ for some $i\in\{1,\dots,n\}$ and $y = 0w_i\bot^\omega$ for some $w_i\in \lang(\aut{D}_i)$.

    Let us prove that $\bigcap_i \lang(\aut{D}_i) \neq \varnothing$ if and only if there exists an optimal strategy from $v_0$ for player~1.
    First, by a direct application of the definition of $\val_1(v_0)$ and the fact that the \DFAs{} $\aut{D}_i$ have non-empty languages, we get $\val_1(v_0) = \{ i\bot^\omega\mid i\in\{1,\dots,n\}\}$. Second, a strategy (i.e., a word $z\in\Sigma'^\omega$) is optimal if, by definition, for all $i\bot^\omega\in \val_1(v_0)$, $i\bot^\omega\R z$. Therefore, by definition of $\R$, such a strategy exists if and only if there exists $w\in\Sigma^*$ such that $z = 0w\bot^\omega$ and $w\in \bigcap_i \lang(\aut{D}_i)$. 

    \newcommand{\vp}[2]{\begin{pmatrix}#1\\#2\end{pmatrix}}
    \newcommand{\vps}[2]{(\begin{smallmatrix}#1\\#2\end{smallmatrix})}
    It remains to show that $\R$ is accepted by a \DPW{} $\aut{A}$ which can be constructed in polynomial time. This follows easily by observing that the language representation of $\R$ consists of words of the form
    \[
    \vp{i}{0}\vp{\bot}{a_1}\vp{\bot}{a_2}\ldots\vp{\bot}{a_k}\vp{\bot}{\bot}^\omega\text{ such that $i\in\{1,\dots,n\}$ and $a_1a_2\dots a_k\in \lang(\aut{D}_i)$.}
    \]
    So, the first letter that $\aut{A}$ is allowed to read must be of the form $\vps{i}{0}$ for any $i\in\{1,\dots,n\}$ (otherwise it leads to a rejecting sink state). Then, the subsequent letters must be of the form $\vps{\bot}{a}$ for $a\in\Sigma$, and the \DPW{} simulates the \DFA{} $\aut{D}_i$ on the $\Sigma$-projection. Eventually, it checks that the infinite sequence ends with $\vps{\bot}{\bot}^\omega$, and that the sequence of $\Sigma$-letters is accepted by $\aut{D}_i$. So, the size of $\aut{A}$ is the sum of the sizes $|\aut{D}_i|$ and some constant. 

    Note that the arena $\arena$ has only two vertices but $n+|\Sigma|+2$ labels. It is possible to transform it into a proper arena (without edge labels) with $n+|\Sigma|+3$ vertices (the set $\{v_0\} \cup \{u_a \mid a \in \Sigma'\}$).
\end{proof}

\section{APW Construction for Consistent Values}
\label{app:apw-cval-construction}

We fix a game $\game = (\arena, (\R_i)_{i \in \players})$ and an initial vertex $v_0 \in V$, and recall the definition of $\cval_{-i}(v_0)$:
\[
\cval_{-i}(v_0) = \{\pi \in \plays(v_0) \mid \exists \sigma_{-i} \in \Gamma_{-i},\; (\forall \rho \in \plays_{\sigma_{-i}}(v_0),\, \pi \notR_i \rho) \;\land\; \pi \in \plays_{\sigma_{-i}}(v_0) \}.
\]

Let us construct an \APW{} $\aut{B}_i = (S_i, s_{\text{init}}^i, \gamma_i, \beta_i)$ accepting $\cval_{-i}(v_0)$ for each $i \in \players$. As explained in the body, we modify the construction of \cref{thm:values-automatic}. Formally:
\begin{itemize}
    \item $S_i = \{s_{\text{init}}^i\} \cup (V \times V \times \{\bot,\top\} \times Q_i)$;
    \item $\beta_i(s_{\text{init}}^i) = 0$, $\beta_i(v,u, \top, q) = \beta_i(v, u, \bot, q) = \alpha_i(q) + 1$ (since we test $\notR_i$);
    \item $\gamma_i(s_{\text{init}}^i, v_0) = (v_0,v_0, \top, q_0^i)$;
    \item From state $(v,v,\top,q)$ and letter $v' \in V$ such that $(v,v') \in E$, with $q' = \delta_i(q,(v,v))$:
    \[
    \gamma_i((v,v,\top,q), v') =
    \begin{cases}
        (v',v',\top,q') & \text{if } v \in V_{-i}, \\
        (v',v',\top,q') \land \bigwedge_{\substack{(v, u') \in E \\ u' \neq v'}} (v', u', \bot, q') & \text{if } v \in V_i;
    \end{cases}
    \]
    \item From state $(v,u,\bot,q)$ and letter $v' \in V$ such that $(v,v') \in E$, with $q' = \delta_i(q,(v,u))$:
    \[
    \gamma_i((v, u, \bot, q), v') =
    \begin{cases}
        \bigvee_{(u, u') \in E} (v', u', \bot, q') & \text{if } u \in V_{-i}, \\
        \bigwedge_{(u, u') \in E} (v', u', \bot, q') & \text{if } u \in V_i.
    \end{cases}
    \]
\end{itemize}

An accepting run tree of $\aut{B}_i$ on $\pi \in V^\omega$ ensures that $\pi \in \plays(v_0)$ and defines a coalition strategy $\tau_{-i} \in \Gamma_{-i}$ that follows $\pi$ along all prefixes of $\pi$ (the single branch with $b = \top$) and guarantees along all branches (both $b = \top$ and $b = \bot$) that $\pi \notR_i \rho$. Conversely, any strategy $\tau_{-i}$ witnessing $\pi \in \cval_{-i}(v_0)$ induces an accepting run tree of $\aut{B}_i$ on $\pi$. Thus $\lang(\aut{B}_i) = \cval_{-i}(v_0)$.
The \APW{} $\aut{B}_i$ has $\bigO{|V|^2\cdot |\aut{A}_i|}$ states and a priority function of index $|\alpha_i|+1$.

\section{Proof of \texorpdfstring{\cref{thm:crs-pspace-ncrs-undec}}{Theorem~\ref{thm:crs-pspace-ncrs-undec}}}
\label{app:rational-synthesis}

\thmrs*

\medskip
Both results are established via separate proofs.

\begin{proof}[Proof (First problem).]
    Solving the CRS problem is very close to solving the NE threshold problem. First, we need a counterpart of \cref{thm:apw-for-ne-outcome} adapted to $0$-fixed NEs. By the exact same proof as in \cref{thm:apw-for-ne-outcome}, $\rho$ is the outcome of a $0$-fixed NE from $v$ if and only if $\rho \in \bigcap_{i \in \players \setminus \{0\}}\cval_{-i}(v)$ (player~$0$ is excluded from the intersection because we only consider deviations from the environment; for each $i \in \players \ssetminus \{0\}$, the witnessing strategy $\tau_{-i}$ of coalition $-i = \players \ssetminus \{i\}$ is consistent with $\rho$, so its restriction to player~$0$ agrees with the strategy $\sigma_0$ that follows $\rho$. In other words, player~$0$ must also participate in the punishment strategy of each coalition, and a single strategy $\sigma_0$ can implement this behavior as the game is turn-based). Hence, as in \cref{thm:apw-for-ne-outcome}, there exists a polynomial-size \APW{} accepting the set $0\NEout(v)$ of outcomes of all $0$-fixed NEs from $v$. Second, as done in the proof of \cref{thm:NEthreshold}, we restrict $0\NEout(v)$ to $L_0 = \{\rho \in V^\omega \mid \pi \R_0 \rho\}$, so that there exists a solution to the CRS instance if and only if $0\NEout(v) \cap L_0 \neq \varnothing$. The latter intersection is again accepted by a polynomial-size \APW{} whose non-emptiness can be checked in polynomial space.

    The \pspaceHard{}ness follows via a reduction from the NE existence problem (see \cref{thm:NEthreshold}). Let $\game = (\arena, (\R_i)_{i \in \players})$ be a game. We construct a new game $\game' = (\arena, (\R_i)_{i \in \players \cup \{0\}})$, with the same arena $\arena$ but with a new player, player~$0$, who controls no vertex. We define the $\omega$-automatic relation $\mathord{\R_0} = V^\omega \times V^\omega$, and set the threshold $\pi$ to be an arbitrary lasso in the arena. Consequently, any NE in $\game$ is also a $\sigma_0$-fixed NE in $\game'$ (for $\sigma_0$ arbitrary as $V_0 = \varnothing$), and the threshold constraint over $\R_0$ is trivially satisfied by all plays. It follows that the CRS problem is \pspaceHard{}.
\end{proof}

\begin{proof}[Proof (Second problem).]
    The proof for the undecidability of the NCRS problem is a reduction from the \emph{recurrent infinite Post correspondence problem} ($\omega$-rPCP), a variant of the infinite Post correspondence problem~\cite{Gire86}, known to be $\Sigma_1^1$-complete~\cite{Harel86}, thus undecidable. A sketch of proof containing the beginning of the construction is given in \cref{section:rational-synthesis}. We detail here the formal construction and correctness establishing the $\Sigma_1^1$-hardness. The $\omega$-rPCP problem is to decide, given two alphabets $A, B$, two morphisms $\varphi_L,\varphi_R: A \to B^+$, and a letter $\bar a \in A$, whether there exists an infinite word $w \in A^\omega$ such that $\varphi_L(w) = \varphi_R(w)$ and $\bar a$ appears infinitely often in $w$. We show how to construct a two-player game $\game = (\mathsf{Arn}, (\R_0,\R_1))$ with initial vertex $v$ and a lasso play $\pi$ in $\mathsf{Arn}$, such that there exists a solution to the $\omega$-rPCP instance if and only if there exists a solution to the NCRS instance in $\game$ for initial vertex $v$ and threshold $\pi$.
    
    \subparagraph*{Arena} Formally, we define the arena $\mathsf{Arn} = (V,E,\{0,1\},(V_0,V_1))$ as follows. Let $\varphi_i^j(a)$ denote the $j$-th letter of $\varphi_i(a)$, with $i \in \{L,R\}$, $a \in A$, and $1 \leq j \leq |\varphi_i(a)|$. Then:
    \begin{itemize}
        \item $V = \{\#\} \cup A \cup \{\varphi_i^j(a) \mid i \in \{L,R\}, a \in A, 1 \leq j \leq |\varphi_i(a)|\} \cup \{L_a,R_a \mid a \in A\} \cup B$,
        \item $V_1 = A$, and $V_0 = V \ssetminus V_1$,
        \item for every $a \in A$, $i \in \{L,R\}$, and $b,b' \in B$, the edges are:
        \begin{itemize}
            \item $(\#,a)$, $(a,\#)$,
            \item $(a,L_a)$, $(a,R_a)$,
            \item $(L_a,\varphi_L^1(a))$, $(R_a,\varphi_R^1(a))$,
            \item $(\varphi_i^j(a), \varphi_i^{j+1}(a))$, for $1 \leq j < |\varphi_i(a)|$,
            \item $(\varphi_i^{n}(a),b)$, for $n = |\varphi_i(a)|$,
            \item $(b,b')$.
        \end{itemize}
    \end{itemize}
    For clarity and to avoid heavy notations, in the definition of $V$, we assume that vertices for letters of $\varphi_i(a)$ and letters of $B$ are distinct, even though $\varphi_i(a) \in B^+$. For example, if $B=\{a,b\}$ and $\varphi_L(x)=aba$ for some $x\in A$, we assume that the three occurrences of $a$ (one in $B$ and two in $\varphi_L(x)$) yield three distinct vertices.
    
    \subparagraph*{Preference relations} Before defining the preference relations $\R_0$ and $\R_1$, let us study the behavior of player~$1$ when player~$0$ fixes a strategy $\sigma_0$. Given $\sigma_0$, we can see all the plays consistent with $\sigma_0$ as an infinite tree as illustrated in \cref{fig:arena-unravelling-undecidability-ncrs}. Each vertex of player~$1$ has the same three choices as presented above. The central branch corresponds to the (first) choice of a word $a = a_1a_2\dots \in A^\omega$. Whenever player~$1$ decides to go to $L_{a_k}$ (resp.\ $R_{a_k}$) after history $\# a_1\ldots \# a_k$ with $k\geq 1$, player~$0$ selects a word $\varphi_L(a_k) \cdot u_k\in B^\omega$ (resp.\ $\varphi_R(a_k) \cdot v_k \in B^\omega$). In this second (resp.\ third) choice, the branch that is followed is called the $k$-left branch (resp.\ $k$-right branch).
    
    From the arena $\mathsf{Arn}$, we construct a game $\game = (\mathsf{Arn}, (\R_0,\R_1))$ where $\mathord{\R_0} = \varnothing$ is the empty relation. The core of the reduction lies in the definition of $\R_1$. For all $a_1,a_2\dots \in A$, $w \in B^\omega$, $k \geq 1$, we define $\R_1$ as follows:
    \begin{enumerate}
        \item\label{item:thmundecidable:cond1} $\# a_1 L_{a_1} w \R_1 \# a_1 R_{a_1} w$,
        \item\label{item:thmundecidable:cond2} $\# a_1\dots \# a_{k+1} L_{a_{k+1}} w \R_1 \# a_1\dots \# a_{k} L_{a_k} \varphi_L(a_k) w$,
        \item\label{item:thmundecidable:cond3} $\# a_1\dots \# a_{k} R_{a_k} \varphi_R(a_k) w \R_1 \# a_1\dots \# a_{k+1} R_{a_{k+1}} w$,
        \item\label{item:thmundecidable:cond4} $\# a_1 \# a_2 \# a_3 \dots \R_1 \# a_1 L_{a_1} w$ if there are infinitely many $i$ such that $a_i = \bar a$.
    \end{enumerate}
    Note that in the definition of each $x \R_1 y$, the same $w$, and $a_1 \dots a_k$ appear in both $x,y$. Moreover, each vertex $\varphi_i^j(a)$ implicitly carries its corresponding letter in $B$, so that the suffix following $L_{a_k}$ (resp.\ $R_{a_k}$) is viewed as a word $w \in B^\omega$.
    
    Given a strategy $\sigma_0$ of player~$0$, $\R_1$ creates an infinite ascending chain of preferences that will later be used in the reduction.
    Returning to~\cref{fig:arena-unravelling-undecidability-ncrs}, \cref{item:thmundecidable:cond1} means that player~$1$ prefers the $1$-right branch to the $1$-left one; \cref{item:thmundecidable:cond2} states, on the left side, that player~$1$ prefers each $k$-left branch to the $(k+1)$-left one; \cref{item:thmundecidable:cond3} states, on the right side, that player~$1$ prefers each $(k+1)$-right branch to the $k$-right one; and \cref{item:thmundecidable:cond4} says that player~$1$ prefers the $1$-left branch to the central one if the central branch sees $\bar a$ infinitely often. In \cref{fig:arena-unravelling-undecidability-ncrs}, we depicted with red dashed arrows the preferences between the $2$-left, $1$-left, $1$-right, and $2$-right branches, assuming that each condition holds. An arrow from a branch $b$ to another branch $b'$ means that $b'$ is preferred to $b$.
    
    The relation $\R_1$ is $\omega$-automatic. Indeed, we can construct a \DPW{} separately for each of \cref{item:thmundecidable:cond1,item:thmundecidable:cond2,item:thmundecidable:cond3,item:thmundecidable:cond4}, then build a generalized \DPW{} for their product, that we finally transform into a \DPW{}. Note that for~\cref{item:thmundecidable:cond2,item:thmundecidable:cond3}, the \DPW{} has to check the suffix equality (with respect to $w$). This is possible because each word $\varphi_i(a)$ has a finite length, and thus, the automaton can buffer the necessary finite sequence of letters from $w$ in its state space to verify the match.
    
    Finally, to complete the reduction, we use any lasso $\pi$ starting with $\#$ as threshold. Let us now prove that the reduction is correct. In particular, we show that there exists a solution to the $\omega$-rPCP instance if and only if there exists a solution to the NCRS instance.
    
    \subparagraph*{Correctness: $\omega$-rPCP $\rightarrow$ non-cooperative rational synthesis} We first assume that there exists a solution $w = a_1a_2\dots \in A^\omega$ to the $\omega$-rPCP instance. We must construct a strategy $\sigma_0$ for player~$0$ such that for all $\sigma_0$-fixed NEs $\bsigma$, we have $\pi \R_0 {\outcomefrom{\bsigma}{\#}}$. As $\R_0$ is the empty relation, we must show that there is no $\sigma_0$-fixed NE. We define the strategy $\sigma_0$ from $w =a_1a_2 \ldots $ as follows:
    \begin{itemize}
        \item player~$0$ generates the sequence $w$, i.e., $\sigma_0(\#) = a_1$ and $\sigma_0(\# a_1\dots a_k \#) = a_{k+1}$ for all $k \geq 1$;
        \item whenever player~1 chooses the edge $(a_k,L_{a_k})$ (resp.\ $(a_k,R_{a_k})$), $\sigma_0$ generates the suffix $\varphi_L(w_{\geq k}):= \varphi_L(a_ka_{k+1}\dots)$ (resp.\ $\varphi_R(w_{\geq k}):=\varphi_R(a_ka_{k+1}\dots)$).
    \end{itemize} 
    
    Let us show that there is no $\sigma_0$-fixed NE. As $w$ is a solution to the given $\omega$-rPCP instance, we have $\varphi_L(w) = \varphi_R(w)$. Let $\sigma_1$ be any strategy of player~$1$ and consider the infinite tree of \cref{fig:arena-unravelling-undecidability-ncrs}.
    \begin{itemize}
        \item If $\outcomefrom{\sigma_0,\sigma_1}{\#}=\#a_1\#a_2\dots$ is the middle branch, then a deviation to the 1-left branch yields the outcome $\# a_1L_{a_1}\varphi_L(w)$ which is profitable for player~$1$ by~\cref{item:thmundecidable:cond4} as $\bar a$ occurs infinitely often in $w$.
        \item If $\outcomefrom{\sigma_0,\sigma_1}{\#}=\#a_1\dots \#a_kR_{a_k}\varphi_R(w_{\geq k})$ is the $k$-right branch, $k \geq 1$, then a deviation to the $(k+1)$-right branch yields the outcome $\#a_1\dots\#a_k\# a_{k+1}R_{a_{k+1}}\varphi_R(w_{\geq k+1})$ which is profitable by~\cref{item:thmundecidable:cond3}.
        \item If $\outcomefrom{\sigma_0,\sigma_1}{\#}=\#a_1\dots \#a_kL_{a_k}\varphi_L(w_{\geq k})$ is the $k$-left branch, $k \geq 2$, then a deviation to the $(k-1)$-left branch yields the outcome 
        $\#a_1\dots \#a_{k-1}L_{a_{k-1}}\varphi_L(w_{\geq k-1})$ which is
        profitable by~\cref{item:thmundecidable:cond2}. 
        \item If $\outcomefrom{\sigma_0,\sigma_1}{\#}=\#a_1L_{a_1}\varphi_L(w)$ is the $1$-left branch, then a deviation to the $1$-right branch yields
        the outcome $\#a_1R_{a_1}\varphi_R(w)$ which is profitable by~\cref{item:thmundecidable:cond1}.
    \end{itemize}
    
    In every case, player~$1$ has a profitable deviation. Thus, no $\sigma_0$-fixed NE exists, and $\sigma_0$ is a solution to the NCRS instance.
    
    \medskip
    
    \subparagraph*{Correctness: non-cooperative rational synthesis $\rightarrow$ $\omega$-rPCP} Conversely, let us suppose that $\sigma_0$ is a solution to the NCRS instance. As $\R_0$ is the empty relation, no outcome for $\sigma_0$ is preferred to $\pi$, and therefore, there exists no $\sigma_0$-fixed NE. Let $w = a_1a_2\dots$ be the sequence of letters chosen by $\sigma_0$ along the middle branch, see~\cref{fig:arena-unravelling-undecidability-ncrs}. Let us prove that $\bar a$ occurs infinitely often in $w$ and $\varphi_L(w) = \varphi_R(w)$, that is, $w$ is a solution to the $\omega$-rPCP instance.
    
    Since there is no $\sigma_0$-fixed NE, for every strategy $\sigma_1$ of player~$1$, there always exists a profitable deviation. Consider a strategy $\sigma_1$ such that $\outcomefrom{\sigma_0,\sigma_1}{\#}$ is the middle branch, then there is a profitable deviation. The only profitable deviation is the $1$-left branch by~\cref{item:thmundecidable:cond4}, showing that $\bar a$ occurs infinitely often in $w$. Moreover, consider a strategy $\sigma_1$ such that $\outcomefrom{\sigma_0,\sigma_1}{\#}$ is the $k$-right branch $\# a_1 \dots \# a_{k}R_{a_{k}} \varphi_R(a_k) v_k$ for some $k \geq 1$. The only profitable deviation for this outcome is the $(k+1)$-right branch $\# a_1 \dots \# a_{k+1}R_{a_{k+1}} \varphi_R(a_{k+1}) v_{k+1}$ given by~\cref{item:thmundecidable:cond3}. It follows that $v_k = \varphi_R(a_{k+1}) v_{k+1}$, and this holds for all $k\geq 1$. Hence, we have
    \begin{align*}
        \varphi_R(a_1) v_1 &= \varphi_R(a_1a_2) v_2 = \varphi_R(a_1a_2a_3) v_3 = \ldots = \varphi_R(a_1a_2a_3\dots).
    \end{align*}
    Similarly, considering strategies $\sigma_1$ whose outcomes are the $(k+1)$-left branches, by~\cref{item:thmundecidable:cond2}, we get $u_k = \varphi_L(a_{k+1}) u_{k+1}$ for all $k \geq 1$, and thus
    \begin{align*}
        \varphi_L(a_1) u_1 &= \varphi_L(a_1a_2) u_2 = \varphi_L(a_1a_2a_3) u_3 = \ldots = \varphi_L(a_1a_2a_3\dots).
    \end{align*}
    Finally, considering the strategy $\sigma_1$ with outcome equal to the 1-left branch, we get by~\cref{item:thmundecidable:cond1} that $\varphi_L(a_1) u_1 = \varphi_R(a_1) v_1$. 
    Therefore, $\varphi_L(a_1a_2a_3\dots) = \varphi_R(a_1a_2a_3\dots)$ showing that $w$ is a solution to the $\omega$-rPCP instance.
\end{proof}

\begin{remark}
    Finally, note that the reduction above can be trivially modified so that every strategy $\sigma_0$ of player~$0$ admits at least one $\sigma_0$-fixed NE response (avoiding vacuous satisfaction): we add a new initial vertex $v_{\text{init}} \in V_1$, and vertices $v_c,\top,\bot \in V_0$ with edges $(v_{\text{init}},\#)$, $(v_{\text{init}},v_c)$, $(v_c,\top)$, $(v_c,\bot)$, $(\top,\top)$ and $(\bot,\bot)$. Let $\pi_\bot = v_{\text{init}}v_c\bot^\omega$ and $\pi_\top = v_{\text{init}}v_c\top^\omega$, we set the threshold to $\pi_\bot$, define $\mathord{\R_0} = \{(\pi_\bot, \pi_\top)\}$, and keep $\pi_\bot$ and $\pi_\top$ incomparable with all plays under $\R_1$ (where $\R_1$ is adapted to consider the new initial vertex $v_{\text{init}}$). With this modified game and threshold, a strategy $\sigma_0$ solves the NCRS instance if and only if $\sigma_0(v_{\text{init}}v_c) = \top$ and no $\sigma_0$-fixed NE exists in the part starting from $\#$. Moreover, player~$1$ moving from $v_{\text{init}}$ to $v_c$ is a $\sigma_0$-fixed NE response.
    \lipicsEnd
\end{remark}

\end{document}